\newif\ifarxiv\arxivtrue%

\documentclass[a4paper,UKenglish,numberwithinsect]{lipics-v2019}
\nolinenumbers%

\RequirePackage{silence}
\title{Tree Automata as Algebras: Minimisation~and~Determinisation}
\author{Gerco van Heerdt}%
{University College London, United Kingdom}%
{gerco.heerdt@ucl.ac.uk}%
{https://orcid.org/0000-0003-0669-6865}
{}
\author{Tobias Kapp\'e}%
{University College London, United Kingdom}%
{tkappe@cs.ucl.ac.uk}%
{https://orcid.org/0000-0002-6068-880X}
{}
\author{Jurriaan Rot}%
{University College London, United Kingdom and Radboud University, The Netherlands}%
{jrot@cs.ru.nl}%
{}%
{}
\author{Matteo Sammartino}%
{University College London, United Kingdom}%
{m.sammartino@ucl.ac.uk}%
{https://orcid.org/0000-0003-1456-2242}
{}
\author{Alexandra Silva}%
{University College London, United Kingdom}%
{alexandra.silva@ucl.ac.uk}%
{https://orcid.org/0000-0001-5014-9784}
{}

\funding{This work was partially supported by the ERC Starting Grant ProFoundNet (grant code 679127), a Leverhulme Prize (PLP--2016--129) and a Marie Curie Fellowship (grant code 795119).}
\authorrunning{G. van Heerdt, T. Kapp\'e, J. Rot, M. Sammartino, A. Silva}

\Copyright{G. van Heerdt, T. Kapp\'e, J. Rot, M. Sammartino, A. Silva}

\ccsdesc[500]{Theory of Computation~Formal languages and automata theory}

\keywords{tree automata, algebras, minimisation, determinisation, Nerode equivalence}
\usepackage[dvipsnames,usenames]{xcolor}

\usepackage[textsize=tiny]{todonotes}
\usepackage{tikz-cd}
\usepackage{xspace}
\usepackage{microtype}

\newcommand{\arity}{\mathsf{ar}}
\newcommand{\aut}{\mathcal{A}}
\newcommand{\tset}{\mathcal{T}}
\newcommand{\lang}{\mathcal{L}}
\newcommand{\freealg}[1]{#1^{\diamond}}

\newcommand{\Alg}{\mathsf{Alg}}

\newcommand{\Kl}{\mathcal{K}{\kern-.4ex}\ell}
\newcommand{\EM}{\mathcal{E}{\kern-.4ex}\mathcal{M}}
\newcommand{\Set}{\mathsf{Set}}

\newcommand{\Nom}{\mathsf{Nom}}
\newcommand{\C}{\mathsf{C}}

\newcommand{\two}{\mathsf{2}}
\newcommand{\one}{\mathsf{1}}

\newcommand{\epi}{\mathcal{E}}
\newcommand{\mono}{\mathcal{M}}
\newcommand{\pow}{\mathcal{P}}
\newcommand{\mult}{\mathcal{M}}

\newcommand{\fpow}{\pow_{\mathsf{f}}}
\newcommand{\sring}{\mathbb{S}}
\newcommand{\fld}{\mathbb{F}}
\newcommand{\fspow}{\pow_{\omega}}
\newcommand{\atoms}{\mathbb{A}}
\newcommand{\Sym}{\mathsf{Sym}}
\newcommand{\id}{\mathsf{id}}

\newcommand{\Quot}{\mathsf{Quot}}
\newcommand{\gfp}{\mathsf{gfp}}
\newcommand{\supp}{\mathsf{supp}}
\newcommand{\reach}{\mathsf{rch}}

\newcommand{\img}{\mathsf{img}}

\newcommand{\todot}{%
  \mathrel{\ooalign{\hfil$\vcenter{
    \hbox{$\scriptscriptstyle\bullet$}}$\hfil\cr$\to$\cr}
  }%
}

\newcommand{\arkl}[2]{\ar{#1}{#2}\ar[phantom]{#1}{\mbox{\scriptsize$\bullet$}}}
\newcommand{\arkls}[2]{\ar{#1}[swap]{#2}\ar[phantom]{#1}{\mbox{\scriptsize$\bullet$}}}
\newcommand{\klex}[1]{\widehat{#1}}

\usepackage{mathtools}

\usepackage{mathpartir}

\newcommand*{\circled}[1]{{\tikz[baseline={(X.base)},scale=0.6]\node(X)[draw,shape=circle,scale=0.6,inner sep=2pt]{{#1}};}}

\definecolor{almond}{rgb}{0.94, 0.87, 0.8}

\usepackage{thmtools}
\usepackage{thm-restate}
\usepackage{thm-patch}

\theoremstyle{definition}
\newtheorem{assumption}[theorem]{Assumption}

\begin{document}

\maketitle

\begin{abstract}
We study a categorical generalisation of tree automata, as $\Sigma$-algebras for a fixed endofunctor $\Sigma$ endowed with initial and final states.
Under mild assumptions about the base category, we present a general minimisation algorithm for these automata.
We then build upon and extend an existing generalisation of the Nerode equivalence to a categorical setting and relate it to the existence of minimal automata.
Finally, we show that generalised types of side-effects, such as non-determinism, can be captured by this categorical framework, leading to a general determinisation procedure.
\end{abstract}

\section{Introduction}

Automata have been extensively studied using category theory, both from an algebraic and a coalgebraic perspective~\cite{Holcombe:1982:AAT:539040,AT89,DBLP:conf/concur/Rutten98,Pin}.
Categorical insights have enabled the development of generic algorithms for minimisation~\cite{DBLP:conf/fossacs/AdamekBHKMS12}, determinisation~\cite{DBLP:journals/corr/abs-1302-1046}, and equivalence checking~\cite{BonchiP15}.

A fruitful line of work has focussed on characterising the semantics of different types of automata as final coalgebras.
The final coalgebra contains unique representatives of behaviour, and
 the existence of a minimal automaton can be formalised by a suitable factorisation of the map from a given automaton into the final coalgebra.
Algorithms to compute the minimal automaton can be devised based on the final sequence, which yields procedures resembling classical partition refinement~\cite{DBLP:conf/ifipTCS/KonigK14,DorschMSW17}.
Unfortunately, bottom-up tree automata do not fit the abstract framework of final coalgebras.\footnote{The language semantics of top-down tree automata represented as coalgebras is given in~\cite{KlinR16}, based on a transformation to bottom-up tree automata. In this paper, we focus on bottom-up automata only.}
This impeded the application of abstract algorithms for minimisation, determinisation, and equivalence.
We embrace the categorical \emph{algebraic} view on automata due to Arbib and Manes~\cite{arbib1974} to study bottom-up tree automata (Section~\ref{sec:tree-aut}).
This algebraic approach is also treated in detail by Ad\'{a}mek and Trnkov\'{a}~\cite{AT89}, who, among other results, give conditions under which minimal realisations exist (see also~\cite{ADAMEK1977281}) and constructions to determinise partial and non-deterministic bottom-up tree automata.
However, generic \emph{algorithms} for minimisation, and a more abstract and uniform picture of determinisation, have not been studied in this context.

The contributions of this paper are three-fold. First, we explore the notion of \emph{cobase} to devise an iterative construction for minimising tree automata, at the abstract level of algebras, resembling partition refinement (Section~\ref{sec:minimisation}). The notion of cobase is dual to that of base~\cite{alwin}, which plays a key role in reachability of coalgebras~\cite{reachability19,BKR19} and therefore in minimisation of automata.
Second, we study a different characterisation of minimality via the Nerode equivalence, again based on work of Arbib and Manes~\cite{arbib1974}, and provide a generalisation using monads that allows to treat automata with equations (Section~\ref{sec:nerode}).
Third, we extend bottom-up tree automata to algebras in the Kleisli category of a monad, which enables us to study tree automata enriched with side-effects and derive an associated determinisation procedure (Section~\ref{sec:determinisation}). We demonstrate the generality of our approach by applying it both to classical examples---deterministic, non-deterministic, and multiplicity/weighted tree automata---and to a novel kind of tree automata, namely \emph{nominal} tree automata.

\section{Preliminaries}%
\label{sec:prelim}

We assume basic knowledge of category theory. Throughout this paper, we fix a category $\C$.

\subparagraph*{Monads}
A \emph{monad} on $\C$ is a triple $(T,\eta,\mu)$ consisting of an endofunctor $T$ on $\C$ and two natural transformations: a \emph{unit} $\eta\colon \mathsf{Id} \Rightarrow T$ and a \emph{multiplication} $\mu\colon T^2 \Rightarrow T$, which satisfy the compatibility laws $\mu \circ \eta_T = \id_T = \mu \circ T\eta$ and $\mu \circ \mu_T = \mu \circ T\mu$.

\begin{example}
The triple $(\fpow, \{-\}, \bigcup)$ is a monad on $\Set$, where $\fpow$ is the finite powerset functor
, $\{-\}$ is the singleton operation, and $\bigcup$ is union of sets. Another example is the \emph{multiplicity monad}
$(\mult_\fld, e, m)$ for a field $\mathbb F$, where $\mult_\fld$ is the functor sending a set $X$ to
$\mult_\fld X = \{\varphi \colon X \to \fld \mid \varphi \; \text{has finite support}\}$. An element $\varphi$ can be seen as a formal finite sum $\sum_i s_i x_i$, where each $x_i$ has multiplicity $s_i$.
 The unit $e$ sends $x$ to $1x$ and the multiplication is $m_X(\sum_i s_i \varphi_i)(x) = \sum_i s_i \cdot \varphi_i(x)$, where $\cdot$ is the field multiplication.
\end{example}

\subparagraph*{Algebras and Varietors}

We fix a functor $\Sigma: \C \to \C$ and write $\mathsf{Alg}(\Sigma)$ for the category of $\Sigma$-algebras.
Throughout this paper, we assume that $\Sigma$ is a \emph{varietor}~\cite{AT89}, i.e., that the forgetful functor $U\colon \mathsf{Alg}(\Sigma) \to \C$ admits a left adjoint $F \colon \C \to \mathsf{Alg}(\Sigma)$. The varietor $\Sigma$ induces a monad $(\Sigma^\diamond, \eta, \mu)$ on $\C$, where $\Sigma^\diamond = UF$.
Given an object $X$ of $\C$, we refer to $FX = (\Sigma^\diamond X, \alpha_X)$ as the \emph{free $\Sigma$-algebra} over $X$.
The free $\Sigma$-algebra satisfies the following: for every $\Sigma$-algebra $Q$ and every morphism $x \colon X \to UQ$ of $\C$, there is a unique $\Sigma$-algebra morphism $x^\sharp \colon (\freealg{\Sigma}X,\alpha_X) \to Q$ with $U(x^\sharp) \circ \eta_X = x$.

\begin{example}
A functor is \emph{finitary} if it preserves filtered colimits. Any finitary $\Set$ functor is a varietor: free algebras over a set $X$ can be obtained as a colimit of a transfinite sequence~\cite{adamek1974free,Kelly80}.
A \emph{polynomial functor} on $\Set$ is a functor inductively defined by
$
	P := \id \mid A \mid P_1 \times P_2 \mid \coprod_{i \in I} P_i
$
 where $A$ is any constant functor. Polynomial functors are finitary and therefore varietors.
\end{example}

\section{Tree automata, categorically}%
\label{sec:tree-aut}

In this section we start our categorical investigation of (bottom-up) tree automata.
We first discuss a general notion of automaton over an endofunctor $\Sigma$ due to Arbib and Manes~\cite{arbib1974} and then discuss how this notion can be instantiated to obtain various kinds of automata.

\begin{definition}[$\Sigma$-tree automaton]
A \emph{$\Sigma$-tree automaton} over objects $I$ and $O$ in $\C$ is a tuple $(Q, \delta, i, o)$ such that $(Q, \delta)$ is a $\Sigma$-algebra and $i \colon I \to Q$ and $o \colon Q \to O$ are morphisms of $\C$. The objects $I$ and $O$ are referred to as the \emph{input object} and \emph{output object} respectively.
A \emph{homomorphism} from an automaton $(Q,\delta,i,o)$ to an automaton $(Q',\delta',i',o')$ is a $\Sigma$-algebra homomorphism $h \colon Q \rightarrow Q'$ (i.e., $\delta' \circ \Sigma h = h \circ \delta$) such that $h \circ i = i'$ and $o' \circ h = o$.
\end{definition}
Throughout this paper, we fix input and output objects $I$ and $O$ respectively.
If $\Sigma$ is clear from the context we sometimes refer to a
$\Sigma$-tree automaton simply as an automaton.

A \emph{language} (over $\Sigma$) is a morphism $L \colon \freealg{\Sigma} I \rightarrow O$.
In the context of an automaton $\mathcal{A} = (Q, \delta, i, o)$, we can think of $U(i^\sharp) \colon \Sigma^\diamond I \to Q$, induced by the free $\Sigma$-algebra $FI$ on $I$, as the \emph{reachability map}, telling us which state is reached by parsing an element of the free algebra over $I$.
We will write $\reach_{\mathcal{A}}$ (or $\reach$ if $\mathcal{A}$ is obvious) for $U(i^\sharp)$.
The \emph{language of $\mathcal{A}$} is the morphism $\lang(\mathcal{A}) \colon \freealg{\Sigma} I \to O$ in $\C$ given by $\lang(\mathcal{A}) = o \circ \reach_\mathcal{A}$.

\begin{example}[Deterministic bottom-up tree automata]
Let us see how $\Sigma$-tree automata can capture deterministic bottom-up tree automata.
We first recall some basic concepts.

A \emph{ranked alphabet} is a finite set of symbols $\Gamma$, where each $\gamma \in \Gamma$ is equipped with an \emph{arity} $\arity(\gamma) \in \mathbb{N}$.
A \emph{frontier alphabet} is a finite set of symbols $I$.
The set of \emph{$\Gamma$-trees} over $I$, denoted $\tset_\Gamma(I)$, is the smallest set such that $I \subseteq \tset_\Gamma(I)$, and for all $\gamma \in \Gamma$ we have that $t_1,\dots,t_{\arity(\gamma)} \in \tset_\Gamma(I)$ implies $(\gamma, t_1,\dots,t_{\arity(\gamma)}) \in \tset_\Gamma(I)$.
In other words, $\tset_\Gamma(I)$ consists of finite trees with leaves labelled by symbols from $I$ and internal nodes labelled by symbols from $\Gamma$; the number of children of each internal node matches the arity of its label.

A ranked alphabet $\Gamma$ gives rise to a polynomial \emph{signature endofunctor} $\Sigma \colon \Set \to \Set$ given by $\Sigma X = \coprod_{\gamma \in \Gamma} X^{\arity(\gamma)}$.
A \emph{deterministic bottom-up tree automaton} is a $\Sigma$-tree automaton $\aut = (Q, \delta, i, o)$ where
$Q$ is finite, $\Sigma$ is a signature endofunctor, and $O = \two$. Here $Q$ is the set of \emph{states}, $i \colon I \to Q$ is the \emph{initial assignment}, $o \colon Q \to \two$ is the characteristic function of \emph{final} states, and for each $\gamma \in \Gamma$ we have a transition function $\delta_\gamma = \delta \circ \kappa_\gamma \colon Q^{\arity(\gamma)} \to Q$.
The \emph{language} $\lang(\aut)$ is the set of all $\Gamma$-trees $t$ such that $(o \circ \hat{\delta})(t) = 1$, where $\hat{\delta} \colon \tset_\Gamma(I) \to Q$ extends $\delta$ to trees by structural recursion:
\begin{mathpar}
\hat{\delta}(\ell) = i(\ell) \quad (\ell \in I)
\and
\hat{\delta}(\gamma,t_1,\dots,t_k) = \delta_\gamma(\hat{\delta}(t_1),\dots,\hat{\delta}(t_k))
\end{mathpar}
In other words, $\lang(\aut)$ contains the trees that evaluate to a final state.
The map $\hat{\delta}$ above is the transpose $i^\sharp$ in the relevant adjunction between $\Set$ and $\Alg(\Sigma)$, where the left adjoint sends a set $I$ to the $\Sigma$-algebra with carrier $\tset_\Gamma(I)$ and the obvious structure map.
\end{example}

\subsection{Nominal tree automata}%
\label{sec:nom-tree}

To show the versatility of our definition, we instantiate it in the category $\Nom$ of nominal sets and equivariant functions.
This results in a notion of nominal tree automaton---along the lines of nominal automata theory~\cite{BojanczykKL14}---which, as we shall see below, provides a useful model for languages of trees with variables and variable binding.
We first recall some basic notations of nominal set theory~\cite{Pitts13}.
Let $\atoms$ be a countable set of \emph{atoms}, and let $\Sym(\atoms)$ be the associated symmetry group, consisting
of all permutations on $\atoms$.
A \emph{nominal set} is a pair $(X,\cdot)$ of a set $X$ and a function $\cdot \colon \Sym(\atoms) \times X \to X$ forming a left action of
$\Sym(\atoms)$ on $X$.
Each $x \in X$ is required to have \emph{finite support}, i.e., there must exist a finite $A \subseteq \atoms$ such that for all $\pi \in \Sym(\atoms)$, if $\pi$ is equal to $\id_A$ when restricted to $A$, then $\pi \cdot x = x$. The minimal such $A$ is denoted $\supp(x)$, and can be understood as the set of ``free'' names of $x$. Given $x \in X$, its \emph{orbit} is the set $\{ \pi \cdot x \mid \pi \in \Sym(\atoms) \}$. We say that a nominal set $X$ is \emph{orbit-finite} whenever it has finitely many orbits. An \emph{equivariant function} $f \colon (X,\cdot) \to (Y,\cdot)$ is a function $X \to Y$ that respects permutations, i.e., $f(\pi \cdot x) = \pi \cdot f(x)$.

Polynomial functors in $\Nom$ support additional operations~\cite{Fiore12}, such as the \emph{name abstraction} functor $[\atoms] \colon \Nom \to \Nom$, which ``binds'' a name in the support. For instance, if $x \in X$, then $\langle a \rangle x \in [\atoms]X$, with $\supp(\langle a \rangle x) = \supp(x) \setminus \{a\}$. The element $\langle a \rangle x$ should be thought of as an equivalence class \emph{up to $\alpha$-conversion} w.r.t.\ the binder $\langle a \rangle$. We can then define tree automata for parsing trees with binders. Consider for instance $\Sigma_\lambda \colon \Nom \to \Nom$ given by
\[
	\Sigma_\lambda X = \underbrace{X \times X}_{\mathsf{appl}} + \underbrace{[\atoms] X}_{\mathsf{lambda}}
\]
describing the syntax of the $\lambda$-calculus~\cite{GabbayM09}. This functor is finitary~\cite{Fiore12}, which implies the existence of free algebras.
Fixing $I = \atoms$, the carrier of the free $\Sigma_\lambda$-algebra over $I$ consists of parse trees for $\lambda$-terms (up to $\alpha$-conversion) with variables in $\atoms$. We can then define automata parsing these trees as $\mathcal{A} = (Q,\delta,i,o)$, where
\begin{itemize}
	\item $Q$ is a nominal set;
	\item $\delta$ consists of two equivariant functions $\delta_{\mathsf{appl}} \colon Q \times Q \to Q$ and $\delta_{\mathsf{lambda}} \colon [\atoms] Q \to Q$;
	\item $i \colon \atoms \to Q$ is an equivariant function, selecting states for parsing variables;
	\item $o \colon Q \to \two$ is an equivariant characteristic function of final states, which implies that if a state is final, so are all the states in its orbit.
\end{itemize}
The reachability function is the equivariant function given by
\[
	\reach_\mathcal{A}(t) =
	\begin{cases}
		i(t) & \text{if $t \in \atoms$} \\
		\delta_{\mathsf{appl}}(\reach_\mathcal{A}(t_1),\reach_\mathcal{A}(t_2)) & \text{if $t = (t_1,t_2)$} \\
		\delta_{\mathsf{lambda}}(\langle a \rangle \reach_{\mathcal{A}}(t') ) & \text{if $t = \langle a \rangle t'$.}
	\end{cases}
\]
The most interesting case is the last one: in order to parse the $\alpha$-equivalence class $\langle a \rangle t'$, we first parse any tree $t'$ such that $\langle a \rangle t'$ is in the class, and then we take the resulting state up to $\alpha$-conversion w.r.t.\ $\langle a \rangle$.
Note that $\lang(\mathcal{A})$ is equivariant, i.e., invariant under permutations of atoms.
Thus $\mathcal{A}$ recognises $\lambda$-trees up to bijective renamings of variables.

\section{Minimisation}%
\label{sec:minimisation}

In this section we define a construction that
allows to minimise a given tree automaton.
We start with a few basic preliminary notions related to quotients and factorisation systems.

\subparagraph*{Factorisations}
    An \emph{$(\epi,\mono)$-factorisation system} on $\C$ consists of classes of morphisms $\epi$ and $\mono$, closed under composition with isos, such that for every morphism $f$ in $\C$ there exist $e \in \epi$ and $m \in \mono$ with $f = m \circ e$, and we have a unique diagonal fill-in property.

	We list a few properties of factorisation systems.
	First, both $\epi$ and $\mono$ are closed under composition.
	Furthermore, if $g \circ f \in \epi$ and $f \in \epi$, then $g \in \epi$.
	Lastly, if $\epi$ consists of epimorphisms, then it is closed under cointersections, i.e.,
	wide pushouts of epimorphisms~\cite{cats}.
	A functor $\Sigma$ is said to \emph{preserve $\epi$-cointersections} if it preserves
	wide pushouts of epimorphisms in $\epi$.
	In that case, for an epimorphism $e$, if $e \in \epi$ then $\Sigma e$ is again an epimorphism.

\subparagraph*{Quotients}
    Define by $\leq$ the order on morphisms with common domain given by $f \leq g$ iff $\exists h. g = h \circ f$.
    This induces an equivalence relation on such morphisms.
    A \emph{quotient} of an object $X$ is an epimorphism $q \colon  X \twoheadrightarrow X'$ identified up to the equivalence, i.e., an equivalence class.
	We denote by $\Quot(X)$ the class of all quotients of $X$.
	The underlying category $\C$ is said to be \emph{cowellpowered} if $\Quot(X)$ is a set for every $X$.
    In that case, if $\C$ is cocomplete, $\Quot(X)$ forms a complete lattice, with the order given by $\leq$, and the least upper bound (join) given by cointersection.
	We denote by $\Quot_\epi(X)$ the set of quotients of $X$ that are in $\epi$.
	(This is well-defined because $\epi$ is closed under isomorphisms.)

	\begin{assumption}\label{as:min}
Throughout this section, $\C$ is cocomplete and cowellpowered.
Moreover, we fix an $(\epi, \mono)$-factorisation system in $\C$, where $\epi$ contains epimorphisms only.
\end{assumption}

\begin{remark}
The category $\Set$ is cocomplete and cowellpowered, and so is $\Nom$ introduced in Section~\ref{sec:nom-tree}.
In general, the existence of an (epi, strong mono)-factorisation system already follows from $\C$
being cocomplete and cowellpowered~\cite{borceux1994}. Allowing a more general choice of factorisation system
will be useful in Section~\ref{sec:nerode}, where we work with a different $\epi$.
\end{remark}

	Let $(Q, \delta)$ be a $\Sigma$-algebra.
	A \emph{quotient algebra} is a $\Sigma$-algebra $(Q', \delta')$ together with a quotient $q \colon Q \twoheadrightarrow Q'$
	in $\epi$ that is an algebra homomorphism.
    Given a $\Sigma$-tree automaton $(Q, \delta, i, o)$,
a \emph{quotient automaton} is a  $\Sigma$-tree automaton $(Q', \delta', i', o')$
together with a quotient $q \colon Q \twoheadrightarrow Q'$ in $\epi$ that is
a homomorphism of automata.


\begin{definition}[Minimisation]\label{def:minimisation}
	The \emph{minimisation} of a $\Sigma$-tree automaton $(Q,\delta, i, o)$ is a quotient automaton $(Q_m, \delta_m, i_m, o_m)$,
	$q \colon Q \twoheadrightarrow Q_m$, such that for any quotient automaton $(Q',\delta',i', o')$,
	$q' \colon Q \twoheadrightarrow Q'$ of $(Q,\delta,i,o)$
	there exists a (necessarily unique) automaton homomorphism $h \colon Q' \twoheadrightarrow Q_m$
	such that $h \circ q' = q$.
\end{definition}
Minimisation is called \emph{minimal reduction} in~\cite{AT89}.
Note that the morphism $h$ in the definition of minimisation is in $\epi$, since $q'$ and $q$ are.
In the sequel, we sometimes refer to a quotient $q \colon Q \twoheadrightarrow Q_m$
as the minimisation if there exist $\delta_m, i_m, o_m$ turning $(Q_m, \delta_m, i_m, o_m), q$
into the minimisation of $(Q,\delta, i, o)$.

\begin{definition}
	A $\Sigma$-tree automaton $\mathcal{A}$ is said to be \emph{reachable} if the associated reachability map $\reach$ is in $\epi$.
	It is \emph{minimal} if it is reachable and for every reachable $\Sigma$-tree automaton $\mathcal{A}'$ s.t.\ $\lang(\mathcal{A}) = \lang(\mathcal{A}')$ there exists a (necessarily unique) homomorphism from $\mathcal{A}'$ to $\mathcal{A}$.
\end{definition}
The above definition of minimality relies on reachability; a more orthogonal (but equivalent) definition of minimality is explored in Section~\ref{sec:simple}.

We conclude with a few observations on the connection between minimisation and minimality, treated in detail in~\cite{AT89}.
We say $\Sigma$ \emph{admits minimisation} of reachable automata if every reachable automaton over $\Sigma$ has a minimisation.
\begin{lemma}\label{lm:minimal-from-initial}
	An automaton $\mathcal{A}$ is minimal iff it is the minimisation of $(\Sigma^\diamond I, \alpha_I, \eta_I, \mathcal{L}(\mathcal{A}))$.
\end{lemma}
\begin{lemma}\label{lm:minimal-minimisation}
The functor $\Sigma$ admits minimisation of reachable automata if and only if a minimal automaton exists for every language over $\Sigma$.
In that case, if $\mathcal{A}$ is reachable, then the minimisation of $\mathcal{A}$ is minimal.
\end{lemma}
\begin{proof}
	For the equivalence, the implication left to right follows from Lemma~\ref{lm:minimal-from-initial}.
    For the converse, one readily shows that the minimisation of an automaton $\mathcal{A}$ is given by the minimal automaton accepting $\lang(\mathcal{A})$.
	The second statement holds by uniqueness of minimisations.
\end{proof}

\subsection{Minimisation via the cobase}

We show how to compute the minimisation of a given automaton $(Q,\delta,i,o)$ using the so-called \emph{cobase}~\cite{alwin}.
This is the dual of the base, which is used in~\cite{BKR19,reachability19} for reachability of coalgebras.
The cobase allows us to characterise the minimisation as the greatest
fixed point of a certain monotone operator on $\Quot_\epi(Q)$,
which is a complete lattice by Assumption~\ref{as:min}.

\begin{definition}
    Let $f \colon \Sigma X \rightarrow Y$ be a morphism.
	The \emph{($\epi$)-cobase} of $f$ (if it exists) is the greatest
	quotient $q \in \Quot_\epi(X)$ such that there exists a morphism $g$
	with $g \circ \Sigma q = f$.
\end{definition}
A concrete instance of the cobase will be given below in Example~\ref{ex:base-and-theta}.
The cobase can be computed as the join of all quotients satisfying the relevant condition, provided that the functor preserves cointersections.

\begin{theorem}[Existence of cobases]\label{thm:cobase}
	Suppose $\Sigma \colon \C \rightarrow \C$
	preserves $\epi$-cointersections. Then every map $f \colon \Sigma X \rightarrow Y$
	has an $\epi$-cobase, given by the cointersection
	\[
        \bigvee \{q \in \Quot_\epi(X) \mid \exists g. \, g \circ \Sigma q = f\} \,.
	\]
\end{theorem}
\begin{proof}
	For $\epi$ the class of all epis, the dual is shown in~\cite{BKR19,reachability19}.
	The proof goes through in the current, more general setting, using
	that $\epi$ is closed under cointersections.
\end{proof}

\begin{remark}
A $\Set$ functor preserves cointersections iff it is finitary~\cite{AT89}. In particular, this is the case for polynomial functors. For $\Nom$ functors we can use that, in general, a functor preserves cointersections if it is finitary and preserves reflexive coequalisers~\cite{AT89}. These conditions hold for polynomial $\Nom$ functors introduced in Section~\ref{sec:nom-tree}, because they preserve sifted colimits~\cite{KurzP10}, which include filtered colimits and reflexive coequalisers.
\end{remark}

We now define an operator on quotients of the state space of an automaton,
which characterises the minimisation of an automaton and gives a way of computing it.
To this end, given a $\Sigma$-algebra $(Q,\delta)$ and a quotient $q \colon Q \twoheadrightarrow Q' \in \Quot_\epi(Q)$,
define the quotient
$\Theta_\delta(q) \colon Q \twoheadrightarrow \Theta_\delta(Q')$
as the cobase of $q \circ \delta$.
This defines a monotone operator $\Theta_\delta \colon \Quot_\epi(Q) \rightarrow \Quot_\epi(Q)$ that has the following important property (see~\cite{BKR19,reachability19}):
\begin{lemma}\label{lm:theta}
	Suppose $\Sigma$ preserves $\epi$-cointersections.
	For any $\Sigma$-algebra $(Q, \delta)$,
	a quotient $q \colon Q \twoheadrightarrow Q'$ in $\Quot_\epi(Q)$
	satisfies $q \leq \Theta_\delta(q)$ iff there is an
	algebra structure $\delta' \colon \Sigma Q' \rightarrow Q'$
	turning $q$ into an algebra homomorphism.
\end{lemma}
The operator $\Theta_\delta$ allows us to quotient the transition structure of the automaton. In order to obtain the minimal automaton, we incorporate the output map $o \colon Q \rightarrow O$ into the construction of a monotone operator based on $\Theta_\delta$.
For technical convenience, we assume that this map is an element of $\Quot_\epi(Q)$.\footnote{This is not a real
restriction: one can just pre-process the automaton by factorising $o$, i.e., keeping only those outputs actually occurring in the automaton.}
The relevant monotone operator for minimisation is $\Theta_\delta \wedge o$ (where the meet $\wedge$ is taken pointwise in $\Quot_\epi(Q)$).
\begin{example}\label{ex:base-and-theta}
Let $\Sigma \colon \Set \rightarrow \Set$ be a polynomial functor induced by signature $\Gamma$.
We first spell out what the cobase means concretely in this case and then study the operator $\Theta_\delta$ in more detail.
Since $\Sigma$ is an endofunctor on $\Set$, the cobase of a map $f \colon \Sigma X \rightarrow Y$
is the largest quotient $q \in \Quot_\epi(X)$
such that for all $t,t' \in \Sigma X$:
\[
    \text{ if } \Sigma q(t) = \Sigma q(t'), \text{ then } f(t) = f(t') \,.
\]
This means that for every $\gamma \in \Gamma$ with $k = \arity(\gamma)$, and any $x_1, \ldots, x_k, y_1, \ldots y_k$, we have that
\[
    \inferrule{%
        q(x_1) = q(y_1)
        \and \dots \and
        q(x_k) = q(y_k)
    }{%
		f(\kappa_\gamma(x_1, \ldots, x_k)) = f(\kappa_\gamma(y_1, \ldots, y_k))
    }
\]
or equivalently that for all $x_1, \ldots, x_k$ and $x_i'$ with $1 \leq i \leq k$ we have
\[
    \inferrule{%
        q(x_i) = q(x_i')
    }{%
		f(\kappa_\gamma(x_1, \ldots, x_{i-1}, x_i, x_{i+1},\ldots, x_k)) = f(\kappa_\gamma(x_1, \ldots, x_{i-1}, x_i', x_{i+1},\ldots, x_k))
    }
\]

\noindent
Suppose $(Q, \delta, i, o)$ is an automaton.
For $q \in \Quot(Q)$, we have $q \leq \Theta_\delta(q) \wedge o$ iff
\begin{itemize}
\item for all $x,x' \in Q$: if $q(x) = q(x')$, then $o(x) = o(x')$; and
\item for all $\gamma \in \Gamma$ with $k = \arity(\gamma)$, and $x_1, \ldots, x_k$ and $x_i'$ with $1 \leq i \leq k$ we have
\[
    \inferrule{%
        q(x_i) = q(x_i')
    }{%
        q(\delta_\gamma(x_1, \ldots, x_{i-1}, x_i, x_{i+1},\ldots, x_k)) = q(\delta_\gamma(x_1, \ldots, x_{i-1}, x_i', x_{i+1},\ldots, x_k))
    }
\]
\end{itemize}
A partition $q$ with the above two properties is known as a \emph{forward bisimulation}~\cite{HogbergMM09}.
\end{example}

\begin{theorem}\label{thm:gfp}
	Suppose $\Sigma$ preserves $\epi$-cointersections. Let $(Q,\delta,i,o)$ be an automaton,
	where $o \in \Quot_\epi(Q)$. Then $\gfp(\Theta_\delta \wedge o)$ is the
	minimisation of $(Q,\delta,i,o)$.
\end{theorem}
\begin{proof}
	Denote the quotient $\gfp(\Theta_\delta \wedge o)$ by $q_m \colon Q \twoheadrightarrow Q_m$.
	Thus $q_m \leq \Theta_\delta(q_m)$ and $q_m \leq o$, hence (using Lemma~\ref{lm:theta})
	there exist $\delta_m, o_m$ turning $q_m$ into an automaton homomorphism from $(Q,\delta,i,o)$ to $(Q_m, \delta_m,q_m \circ i, o_m)$.
	We show that this is the minimisation of $(Q,\delta,i,o)$.

	To this end, let $(Q', \delta', i', o')$, $q' \colon Q \twoheadrightarrow Q'$ be a quotient automaton of $(Q,\delta, i, o)$.
	By Lemma~\ref{lm:theta} we get $q' \leq \Theta_\delta(q')$, and since $o' \circ q' = o$ we have $q' \leq o$, hence
	$q' \leq \Theta_\delta(q') \wedge o$. Thus $q' \leq \gfp(\Theta_\delta \wedge o)$, i.e., there is a quotient $h \colon Q' \twoheadrightarrow Q_m$
	such that $h \circ q' = q_m$. It only remains to show that $h$ is a homomorphism of automata.
	First, since $q' \in \epi$ and $\Sigma$ preserves $\epi$-cointersections,
	$\Sigma q'$ is an epimorphism. Combined with the fact that $q'$ and $q_m$ are algebra homomorphisms
	and that $h \circ q' = q_m$, it easily follows that $h$ is an algebra homomorphism. To see that it preserves the
	output, we have
	$
	o_m \circ h \circ q' = o_m \circ q_m = o = o' \circ q'
	$;
	hence, since $q'$ is epic, we get $o_m \circ h = o'$.
	For preservation of the input, we have $h \circ i' = h \circ q' \circ i = q_m \circ i$, where
	the first step holds because $q'$ is a homomorphism of automata.
\end{proof}

The above characterisation of minimisation of an automaton $(Q,\delta,i,o)$
gives us two ways of constructing it by standard lattice-theoretic computations.
First, via the Knaster-Tarski theorem, we obtain it as the join of all
post-fixed points of $\Theta_\delta \wedge o$, which, by Lemma~\ref{lm:theta},
amounts to the join of all quotient algebras respecting the output map $o$.
That corresponds to the construction in~\cite{AT89}.
Second, and perhaps most interestingly, we obtain the minimisation of $(Q,\delta,i,o)$
by iterating $\Theta_\delta \wedge o$, starting from the top element $\top$ of
the lattice $\Quot_\epi(Q)$. The latter construction
is analogous to the classical partition refinement algorithm: Starting from $\top$ corresponds to
identifying all states as equivalent (or in other words, starting from the coarsest equivalence class of states). Every iteration step of $\Theta_\delta \wedge o$ splits the states that can be distinguished successively by just outputs, trees of depth $1$, trees of depth $2$, etc.
If the state space is finite, this construction terminates, yielding the
minimisation of the original automaton by Theorem~\ref{thm:gfp}.

\subsection{Simple automata}%
\label{sec:simple}

We defined an automaton to be minimal if it is reachable and satisfies a universal property w.r.t.\ reachable automata accepting the same language.
It is also interesting to ask whether there is another property that, together with reachability, implies minimality, but is not itself dependent on reachability~\cite{AM75}.
Here we propose precisely such a condition.
\begin{definition}
	An automaton $(Q, \delta, i, o)$ is called \emph{simple}
	if for every quotient automaton $(Q', \delta', i', o')$ the associated quotient
	$q \colon Q \twoheadrightarrow Q'$ is an isomorphism.
\end{definition}
The result below asserts that minimal automata are precisely the automata that are simple and reachable.
It can be seen as a refinement (and dual) of~\cite[Theorem 17]{BKR19}, computing the reachable part of a coalgebra.
One of the implications makes use of Theorem~\ref{thm:gfp}, so we assume that $\Sigma$ preserves $\epi$-cointersections.

\begin{proposition}\label{prop:minimisation-simple}
	Suppose $\Sigma$ preserves $\epi$-cointersections.
	Let $(Q,\delta,i,o)$ be an automaton with $o \in \epi$, and let $(Q', \delta',i',o')$, $q \colon Q \twoheadrightarrow Q'$
	be a quotient automaton. Then $(Q', \delta', i', o')$ is simple if and only if
	it is the minimisation of $(Q, \delta, i, o)$.
\end{proposition}
\begin{proof}
By Theorem~\ref{thm:gfp}, the minimisation of $(Q, \delta, i, o)$ exists.
	We denote it by $q_m \colon Q \twoheadrightarrow Q_m$ and its associated automaton structure by $(Q_m, \delta_m, i_m, o_m)$.

	Suppose $(Q', \delta', i', o')$ is simple.
	Since $Q_m$ is the minimisation of $Q$ and $Q'$ is a quotient automaton of $Q$, there exists a homomorphism of automata
	$h \colon Q' \twoheadrightarrow Q_m$. Since $Q'$ is simple, this homomorphism is an iso.

	Conversely, suppose $(Q', \delta', i', o')$ is the minimisation of $(Q, \delta, i, o)$ and consider any quotient automaton $Q''$ of $Q'$, witnessed by
	some $q' \colon Q' \twoheadrightarrow Q''$. Then $Q''$ is also a quotient automaton of $Q$, via
	$q' \circ q$. Because $Q'$ is the minimisation of $Q$, there exists $k \colon Q'' \to Q'$ such that $k \circ q' \circ q = q$.
	Thus $k \circ q' = \id$, using that $q$ is an epi.
	Since $q' \circ k \circ q' = q'$ and $q'$ is an epi as well, we also have $q' \circ k = \id$.
	Hence $q'$ is an iso, as needed.
\end{proof}

\begin{corollary}
	If $\Sigma$ preserves $\epi$-cointersections, then an automaton $\mathcal{A} = (Q,\delta,i,o)$ with $o \in \epi$ is minimal if and only if
	it is simple and reachable.
\end{corollary}
\begin{proof}
    First, suppose that $\mathcal{A}$ is minimal.
	By Lemma~\ref{lm:minimal-from-initial}, we know that $\mathcal{A}$ is the minimisation (and, in particular, a quotient) of $(\Sigma^\diamond I, \alpha_I, \eta_I, \mathcal{L}(\mathcal{A}))$.
	In that case, $\mathcal{A}$ is reachable, and thus, since $o \in \epi$, we have $\mathcal{L}(\mathcal{A}) \in \epi$.
	By Proposition~\ref{prop:minimisation-simple}, we conclude that $\mathcal{A}$ is simple.

    Conversely, let $\mathcal{A}$ be simple and reachable.
    By reachability, $\mathcal{A}$ is a quotient automaton of $(\Sigma^\diamond I, \alpha_I, \eta_I, \mathcal{L}(\mathcal{A}))$; also, $\mathcal{L}(\mathcal{A}) = o \circ \reach \in \epi$.
    Proposition~\ref{prop:minimisation-simple} then tells us that $\mathcal{A}$ is the minimisation of $(\Sigma^\diamond I, \alpha_I, \eta_I, \mathcal{L}(\mathcal{A}))$; by Lemma~\ref{lm:minimal-from-initial} we conclude that $\mathcal{A}$ is minimal.
\end{proof}

\section{Nerode equivalence}%
\label{sec:nerode}

We now show a generalised Nerode equivalence from which the minimal automaton can be constructed.
Most of this section is based upon the work by Arbib and Manes~\cite{arbib1974}, whose construction was further studied and refined by Anderson et al.~\cite{anderson1976} and Ad\'{a}mek and Trnkov\'a~\cite{AT89}.
We make a significant improvement in generality by phrasing the central equivalence definition (Definition~\ref{def:nerode}) in terms of an arbitrary monad, which unlike the previous cited work allows applications to algebras satisfying a fixed set of equations.
A monad generalisation of the Myhill-Nerode theorem appears in~\cite{bojanczyk2015}, which confines itself to categories of sorted sets and does not characterise the equivalence as an object.

The abstract construction in this section does not require the varietor $\Sigma$.
Instead, we focus on the monad $\Sigma^\diamond$ induced by its adjunction and generalise by fixing an arbitrary monad $(T, \eta, \mu)$ in $\C$.
Let $F \dashv U \colon \C \leftrightarrows \EM(T)$ be the adjunction with its category of (Eilenberg-Moore) algebras.
Given a $\C$-morphism $f \colon X \to UY$ for $X$ in $\C$ and $Y$ in $\EM(T)$, we write $f^\sharp \colon FX \to Y$ for its adjoint transpose.
We can then use a generalised version of the automata defined in Section~\ref{sec:tree-aut}.

\begin{definition}[$T$-automaton]
	A \emph{$T$-automaton} is a tuple $(Q, \delta, i, o)$, where $(Q, \delta)$ is a $T$-algebra and $i \colon I \to Q$ and $o \colon Q \to O$ are morphisms in $\C$.
	A \emph{homomorphism} from $(Q, \delta, i, o)$ to $(Q', \delta', i', o')$ is a $T$-algebra homomorphism $h \colon (Q, \delta) \rightarrow (Q', \delta')$ such that $h \circ i = i'$ and $o' \circ h = o$.
\end{definition}

The reachability map of a $T$-automaton $\aut = (Q, \delta, i, o)$ is given by $\reach_\aut = U(i^\sharp) \colon TI \to Q$ and is therefore the unique $T$-algebra homomorphism $(TI, \mu) \to (Q, q)$ preserving initial states, taking $\eta_I \colon I \to TI$ to be the initial state selector of $TI$.
The language of $\aut$ is given by $\lang(\aut) = o \circ \reach_\aut \colon TI \to O$.

The $\Sigma$-tree automata defined in Section~\ref{sec:tree-aut} are recovered using the following fact: the category of $\Sigma$-algebras is isomorphic to $\EM(T)$ for $T$ the free $\Sigma$-algebra monad $\Sigma^\diamond$.

\begin{remark}
In $\Set$, we may even add equations to the signature~\cite[Chapter~VI.8, Theorem~1]{maclane2013}. For $\Nom$, this follows from the treatment of~\cite{KurzP10}, giving a standard universal algebraic presentation of algebras over $\Nom$. Indeed, results in this section apply to nominal tree automata, unless explicitly stated.
For brevity we therefore focus on examples in $\Set$.
\end{remark}

\begin{assumption}
In this section we will need the class $\epi$ to be the reflexive regular epis.\footnote{%
	In a regular category, (reflexive regular epi, mono) forms a factorisation system in the same way (regular epi, mono) does, though one should note that the theory in this section does not actually need a factorisation system; the instantiation of $\epi$ is only invoked to obtain the right notion of reachability.
}
\end{assumption}

The next lemma will be used in proving our main theorems.

\begin{restatable}{lemmarestate}{coequalizerlift}\label{lem:coeqsharp}
	Suppose $T$ maps reflexive coequalisers to epimorphisms.
    If $i \colon B \to UC$ is such that $U(i^\sharp)$ reflexively coequalises $q_1, q_2 \colon A \to TB$ in $\C$, then $i^\sharp$ reflexively coequalises $q_1^\sharp, q_2^\sharp \colon FA \to FB$.
\end{restatable}

Before defining an abstract Nerode equivalence, we recall the classical definition for languages of words.
Given a language $L \colon A^* \to 2$, the equivalence $R \subseteq A^* \times A^*$ is defined as
\[
	R = \{(u, v) \in A^* \times A^* \mid \forall w \in A^*.\,L(uw) = L(vw)\}.
\]
In this setting, $I = 1$ and $O = 2$.
A function $Q \times A \to Q$ corresponds to an algebra for the monad $T = (-) \times A^*$, whose unit and multiplication are defined using the unit and multiplication of the monoid $A^*$.
If $p_1, p_2 \colon R \to A^* \cong 1 \times A^*$ are the projections, we note that $R$ is defined to be the largest relation making the following diagram commute.
\[
	\begin{tikzcd}[column sep=.3cm,row sep=.4cm]
		R \times A^* \ar{rr}{p_2 \times \id} \ar{dd}[swap]{p_1 \times \id} &
			&
			1 \times A^* \times A^* \ar{d}[outer sep=1mm]{\mu} \\
		&
			&
			1 \times A^* \ar{d}[outer sep=1mm]{L} \\
		1 \times A^* \times A^* \ar{r}[outer sep=1mm]{\mu} &
			1 \times A^* \ar{r}[outer sep=1mm]{L} &
			2
	\end{tikzcd}
\]
This leads to an abstract definition, using a limit\footnote{%
	Note that it is not exactly a limit, as the defining property works with cones under $T$.
} to generalise what it means to be maximal.

\begin{definition}[Nerode equivalence]\label{def:nerode}
	Given a language $L \colon TI \to O$ and an object $R$ with morphisms $p_1, p_2 \colon R \to TI$, we say that $(R, p_1, p_2)$ is the \emph{Nerode equivalence} of $L$ if the diagram below on the left commutes and for all objects $S$ with a reflexive pair $q_1, q_2 \colon S \to TI$ such that the diagram in the middle commutes there is a unique morphism $u \colon S \to R$ making the diagram on the right commute.
	\begin{align*}
		\begin{tikzcd}[ampersand replacement=\&,column sep=.6cm,row sep=.4cm]
			TR \ar{rr}{T{p_2}} \ar{dd}[swap]{T{p_1}} \&
				\&
				TTI \ar{d}[outer sep=1mm]{\mu} \\
			\&
				\&
				TI \ar{d}[outer sep=1mm]{L} \\
			TTI \ar{r}{\mu} \&
				TI \ar{r}{L} \&
				O
		\end{tikzcd} &
		&
			\begin{tikzcd}[ampersand replacement=\&,column sep=.6cm,row sep=.4cm]
				TS \ar{rr}{T{q_2}} \ar{dd}[swap]{T{q_1}} \&
					\&
					TTI \ar{d}[outer sep=1mm]{\mu} \\
				\&
					\&
					TI \ar{d}[outer sep=1mm]{L} \\
				TTI \ar{r}{\mu} \&
					TI \ar{r}{L} \&
					O
			\end{tikzcd} &
			&
			\begin{tikzcd}[ampersand replacement=\&]
				\&
					S \ar{ld}[swap]{q_1} \ar[dashed]{d}{u} \ar{rd}{q_2} \\
				TI \&
					R \ar{l}[swap]{p_1} \ar{r}{p_2} \&
					TI
			\end{tikzcd}
	\end{align*}
\end{definition}

To show the versatility of our definition, we briefly explain a different example where the language is a set of words.
This example cannot be recovered from the original definition by Arbib and Manes~\cite{arbib1974}.

\begin{example}[Syntactic congruence]
	Let $T$ be the free monoid or list monad ${(-)}^*$ so that $\EM(T)$ is the category of monoids, $I = A$, and $O = 2$.
	Given a language $L \colon A^* \to 2$, the Nerode equivalence as defined above is then the largest relation $R \subseteq A^* \times A^*$ such that
	\[
		\frac{n \in \mathbb{N} \qquad (u_1, v_1), \ldots, (u_n, v_n) \in R}{L(u_1 \cdots u_n) = L(v_1 \cdots v_n)}.
	\]
	Equivalently, $R$ is the largest relation such that
	\[
		\frac{(u, v) \in R \qquad w, x \in A^*}{L(wux) = L(wvx)},
	\]
	which is precisely the \emph{syntactic congruence} of the language.
\end{example}

We can show that the Nerode equivalence in $\Set$ exists, as long as the monad is finitary.
To define it concretely, we use the following piece of notation.
For any set $X$ and $x \in X$, denote by $1_x \colon 1 \to X$ the constant $x$ function, assuming no ambiguity of the set involved.

\begin{restatable}{propositionrestate}{nerodefin}\label{prop:nerodefin}
	For $\C = \Set$ and $T$ any finitary monad, every language $L \colon TI \to O$ has a Nerode equivalence given by
	\[
		R = \{(u, v) \in TI \times TI \mid L \circ \mu \circ T[\id_{TI}, 1_u] = L \circ \mu \circ T[\id_{TI}, 1_v] \colon T(TI + 1) \to O\}
	\]
	with the corresponding projections $p_1, p_2 \colon R \to TI$.
\end{restatable}

The definition of $R$ above states that $u, v \in TI$ are related iff the elements of $TI$ formed by putting either $u$ or $v$ in any \emph{context} and then applying $\mu$ have the same value under $L$.
A context is an element of $T(TI + 1)$, where $1 = \{\square\}$ denotes a \emph{hole} where either $u$ or $v$ can be plugged in.
In the tree automata literature, such contexts, although restricted to contain a single instance of $\square$, are used in algorithms for minimisation~\cite{HogbergMM09} and learning~\cite{sakakibara1990,drewes2007}.
Unfortunately, the characterisation of Proposition~\ref{prop:nerodefin} does not directly extend to $\Nom$, because the functions $1_x$ are not, in general, equivariant.
We leave this for future work.

Below we show that, under a few mild assumptions, the abstract equivalence is in fact a congruence: it induces a $T$-automaton, which moreover is minimal.
Intuitively, given a language $L \colon TI \to O$ that has a Nerode equivalence, we use the equivalence to quotient the $T$-automaton $(FI, \eta, L)$.
We first need a technical lemma.

\begin{restatable}{lemmarestate}{liftnerode}\label{lem:liftnerode}
	If $\C$ has coproducts, then for any Nerode equivalence $(R, p_1, p_2)$ there exists a unique $T$-algebra structure $u \colon TR \to R$ making $p_1$ and $p_2$ $T$-algebra homomorphisms $(R, u) \to (TI, \mu)$ that have a common section.
\end{restatable}

\begin{theorem}
	If $\C$ has coproducts and reflexive coequalisers and $T$ preserves reflexive coequalisers, then for every language that has a Nerode equivalence there exists a minimal $T$-automaton accepting it.
\end{theorem}
\begin{proof}
	Let $L \colon TI \to O$ be the language with Nerode equivalence $(R, p_1, p_2)$ and $c \colon T \to M$ the coequaliser of $p_1$ and $p_2$ in $\C$.
	By Lemma~\ref{lem:liftnerode} there exists a $T$-algebra structure on $R$ making $p_1$ and $p_2$ $T$-algebra homomorphisms into $(TI, \mu)$ that have a common section.
	Since $T$ preserves reflexive coequalisers, they are lifted by $U$, and we have a morphism $m \colon TM \to M$ making $(M, m)$ a $T$-algebra such that $c$ is a $T$-algebra homomorphism $(TI, \mu) \to (M, m)$.
	Since the diagram below on the left commutes and $c$ coequalises $p_1$ and $p_2$, there is a unique morphism $o_M$ rendering the diagram on the right commutative.
	\begin{equation}\label{eq:minlang}
		\begin{tikzcd}[column sep=.6cm,row sep=.4cm]
			R \ar{rr}{p_2} \ar{rd}{\eta} \ar{dd}[swap]{p_1} \ar[phantom,bend left=5]{rrd}[pos=.6]{\circled{1}} \ar[phantom,bend right=10]{rdd}[pos=.6]{\circled{1}} &
				&
				TI \ar[equal,bend left=25]{dr}{} \ar{d}{\eta} \ar[phantom,bend right=5]{dr}{\circled{2}} \\
			&
				TR \ar{r}{T{p_2}} \ar{d}{T{p_1}} \ar[phantom]{rrdd}{\circled{3}} &
				TTI \ar{r}{\mu} &
				TI \ar{dd}{L} \\
			TI \ar[equal,bend right=25]{dr}{} \ar{r}{\eta} \ar[phantom,bend left=5]{dr}{\circled{2}} &
				TTI \ar{d}{\mu} \\
			&
				TI \ar{rr}{L} &
				&
				O
		\end{tikzcd}
		\qquad
		\begin{gathered}
			\begin{array}{l}
				\circled{1}\text{ naturality of $\eta$} \\
				\circled{2}\text{ monad law} \\
				\circled{3}\text{ Nerode equivalence}
			\end{array} \\[.3cm]
			\begin{tikzcd}[column sep=.6cm,row sep=.3cm]
				TI \ar{r}{c} \ar{rd}[swap]{L} &
					M \ar[dashed]{d}{o_M} \\
				&
					O
			\end{tikzcd}
		\end{gathered}
	\end{equation}
	Choosing $i_M = c \circ \eta \colon I \to M$, we obtain a $T$-automaton $\mathcal{M} = (M, m, i_M, o_M)$.
	Note that $U(i_M^\sharp) = c$, so $c$ is the reachability map of $\mathcal{M}$.
	Hence, we find that $\lang(\mathcal{M}) = L$ by~\eqref{eq:minlang}.
	The morphism $c$ coequalises the reflexive pair $(p_1, p_2)$ by definition, so $\mathcal{M}$ is reachable.

	To see that $\mathcal{M}$ is minimal, consider any reachable $T$-automaton $\aut = (Q, \delta, i, o)$ such that $\lang(\aut) = L$.
	Reachability amounts to the reachability map $\reach \colon TI \to UQ$ being the reflexive coequaliser of a pair of morphisms $q_1, q_2 \colon S \to TI$.
	From commutativity of
	\[
		\begin{tikzcd}[column sep=.6cm,row sep=.4cm]
			TS \ar{rrr}{T{q_2}} \ar{ddd}[swap]{T{q_1}} \ar[phantom]{rd}{\circled{1}} &
				&
				&
				TTI \ar{d}{\mu} \ar{dll}[swap]{T{\reach}} \\
			&
				T{Q} \ar{rd}{\delta} \ar[phantom,bend left=10]{rr}[pos=.7]{\circled{2}} \ar[phantom]{dd}{\circled{2}} &
				&
				TI \ar{dd}{L} \ar{dl}[swap]{\reach} \\
			&
				&
				Q \ar{rd}{o} \ar[phantom]{r}{\circled{3}} \ar[phantom]{d}{\circled{3}} &
				{} \\
			TTI \ar{ruu}{T{\reach}} \ar{r}{\mu} &
				TI \ar{ru}{\reach} \ar{rr}[pos=.65]{L} &
				{} &
				O
		\end{tikzcd}
		\qquad
		\begin{array}{l}
			\circled{1}\text{ $\reach$ coequalises $q_1$ and $q_2$} \\
			\circled{2}\text{ $\reach$ is a $T$-algebra homomorphism} \\
			\circled{3}\text{ $\lang(\aut) = L$}
		\end{array}
	\]
	we obtain by the Nerode equivalence property a unique morphism $v \colon S \to R$ making the diagram below on the left commute.
	\begin{align*}
		\begin{tikzcd}[ampersand replacement=\&]
			\&
				S \ar{ld}[swap]{q_1} \ar[dashed]{d}{v} \ar{rd}{q_2} \\
			TI \&
				R \ar{l}[swap]{p_1} \ar{r}{p_2} \&
				TI
		\end{tikzcd} &
			&
			\begin{tikzcd}[ampersand replacement=\&,column sep=.6cm,row sep=.3cm]
				S \ar{rr}{q_2} \ar{rd}{v} \ar{dd}[swap]{q_1} \&
					\&
					TI \ar{dd}{c} \\
				\&
					R \ar{ur}{p_2} \ar{dl}[swap]{p_1} \\
				TI \ar{rr}{c} \&
					\&
					M
			\end{tikzcd}
	\end{align*}
	Extending this with $c$, the coequaliser of $p_1$ and $p_2$, gives the commutative diagram on the right.
	Recall that $U(i_M^\sharp) = c$.
	We now find
	\[
		i_M^\sharp \circ q_1^\sharp = {(U(i_M^\sharp) \circ q_1)}^\sharp = {(c \circ q_1)}^\sharp = {(c \circ q_2)}^\sharp = {(U(i_M^\sharp) \circ q_2)}^\sharp = i_M^\sharp \circ q_2^\sharp.
	\]
	Here the first and last equality apply a general naturality property of the adjunction.
	Since $\reach = U(i^\sharp)$ is the reflexive coequaliser of $q_1$ and $q_2$, $i^\sharp$ is the reflexive coequaliser of $q_1^\sharp$ and $q_2^\sharp$ by Lemma~\ref{lem:coeqsharp}.
	We then obtain a unique $T$-algebra homomorphism $h \colon (Q, \delta) \to (M, m)$ making the diagram below on the left commute.
	
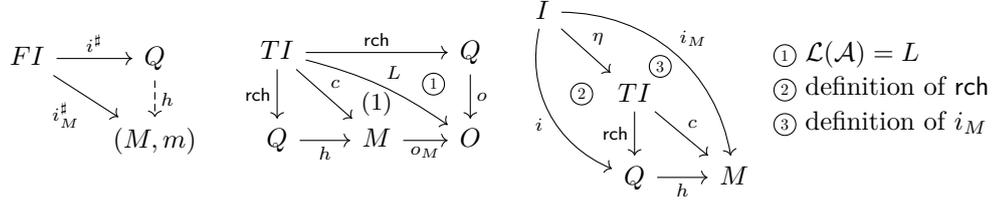
\[
		\begin{gathered}
			\begin{tikzcd}[column sep=.6cm,row sep=.5cm]
				FI \ar{r}{i^\sharp} \ar{rd}[swap]{i_M^\sharp} &
					Q \ar[dashed]{d}{h} \\
				&
					(M, m)
			\end{tikzcd}
		\end{gathered}
		\quad
		\begin{gathered}
			\begin{tikzcd}[column sep=.6cm,row sep=.6cm]
				TI \ar{rr}{\reach} \ar[bend left=10]{rrd}{L} \ar{rd}{c} \ar{d}[swap]{\reach} &
					{} \ar[phantom]{d}[pos=.7]{\text{\small\eqref{eq:minlang}}} \ar[phantom,bend left=20]{dr}[pos=.55]{\circled{1}} &
					Q \ar{d}{o} \\
        	    Q \ar{r}[swap]{h} &
					M \ar{r}[swap]{o_M} &
					O
			\end{tikzcd}
		\end{gathered}
		\quad
		\begin{gathered}
			\begin{tikzcd}[column sep=.6cm,row sep=.6cm]
				I \ar{rd}{\eta} \ar[bend right=40]{rdd}[swap]{i} \ar[bend left=40]{rrdd}{i_M} \ar[phantom,bend right=5]{rdd}[swap]{\circled{2}} \ar[phantom,bend left=25]{rrdd}{\circled{3}} \\
				&
					TI \ar{d}[swap]{\reach} \ar{rd}{c} \\
				&
					Q \ar{r}[swap]{h} &
					M
			\end{tikzcd}
        	\begin{array}{l}
				\circled{1}\text{ $\lang(\aut) = L$} \\
				\circled{2}\text{ definition of $\reach$} \\
				\circled{3}\text{ definition of $i_M$}
        	\end{array}
		\end{gathered}
	\]
	From commutativity of the other diagrams we find $o_M \circ h = o$ (using that $\reach$ is epi) and $h \circ i = i_M$.
	Thus, $h$ is a $T$-automaton homomorphism $\aut \to \mathcal{M}$.
	To see that it is unique, note that any $T$-automaton homomorphism $h' \colon \aut \to \mathcal{M}$ is a $T$-algebra homomorphism $(Q, \delta) \to (M, m)$ such that $h' \circ i = i_M$.
	It is then not hard to see that $h' \circ i^\sharp = {(h' \circ i)}^\sharp = i_M^\sharp$.
	We conclude that $h' = h$ by the uniqueness property of $h$ satisfying $h \circ i^\sharp = i_M^\sharp$.
\end{proof}

\begin{remark}
We briefly discuss the conditions of the above theorem in the specific case of $\C = \Set$ with $T$ a finitary monad.
This includes the setting of tree automata in $\Set$, as a monad on $\Set$ is finitary if and only if $\EM(T)$ is equivalent to the category of algebras for a signature modulo equations.
Proposition~\ref{prop:nerodefin} shows that all Nerode equivalences exist here.
Furthermore, Lack and Rosick{\`y}~\cite{lack2011} observe that an endofunctor on $\Set$ is finitary if and only if it preserves sifted colimits, of which reflexive coequalisers form an instance.
\end{remark}

To conclude this section we show that the converse of the previous theorem also holds, using the existence of kernel pairs rather than coproducts.
We need a technical lemma first.

\begin{restatable}{lemmarestate}{reflexivelift}\label{lem:reflexive}
	If $q_1, q_2 \colon A \to TB$ is a reflexive pair in $\C$, then so is $(q_1^\sharp, q_2^\sharp)$ in $\EM(T)$.
\end{restatable}

\begin{theorem}
	If $\C$ has kernel pairs and reflexive coequalisers and $T$ preserves reflexive coequalisers, then every language that has a minimal $T$-automaton has a Nerode equivalence.
\end{theorem}
\begin{proof}
	Let $\mathcal{M} = (M, \delta_M, i_M, o_M)$ be a minimal $T$-automaton and $p_1, p_2 \colon K \to TI$ the kernel pair of its reachability map $\reach \colon FI \to M$.
	We claim that $K$ together with $p_1$ and $p_2$ forms the Nerode equivalence of $\lang(\mathcal{M})$.
	To see this, note that the diagram below on the left commutes.
	
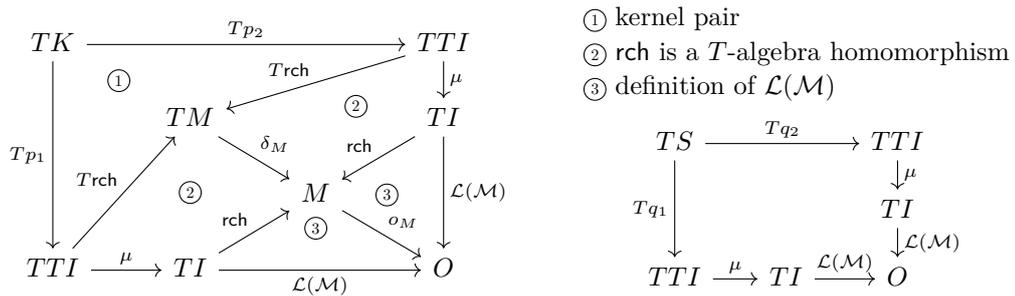
\[
		\begin{tikzcd}[row sep=.5cm]
			TK \ar{rrr}{Tp_2} \ar{ddd}[swap]{Tp_1} \ar[phantom]{rd}{\circled{1}} &
				&
				&
				TTI \ar{d}{\mu} \ar{dll}[swap]{T{\reach}} \\
			&
				TM \ar{rd}{\delta_M} \ar[phantom,bend left=10]{rr}[pos=.7]{\circled{2}} \ar[phantom]{dd}{\circled{2}} &
				&
				TI \ar{dd}{\lang(\mathcal{M})} \ar{dl}[swap]{\reach} \\
			&
				&
				M \ar{rd}{o_M} \ar[phantom]{r}{\circled{3}} \ar[phantom]{d}[pos=.3]{\circled{3}} &
				{} \\
			TTI \ar{ruu}{T{\reach}} \ar{r}{\mu} &
				TI \ar{ru}{\reach} \ar[swap]{rr}{\lang(\mathcal{M})} &
				{} &
				O
		\end{tikzcd}
		\qquad
		\begin{gathered}
		\begin{array}{l}
			\circled{1}\text{ kernel pair} \\
			\circled{2}\text{ $\reach$ is a $T$-algebra homomorphism} \\
			\circled{3}\text{ definition of $\lang(\mathcal{M})$}
		\end{array} \\
		\begin{tikzcd}[column sep=.6cm,row sep=.4cm]
			TS \ar{rr}{T{q_2}} \ar{dd}[swap]{T{q_1}} &
				&
				TTI \ar{d}{\mu} \\
			&
				&
				TI \ar{d}{\lang(\mathcal{M})} \\
			TTI \ar{r}{\mu} &
				TI \ar{r}{\lang(\mathcal{M})} &
				O
		\end{tikzcd}
		\end{gathered}
	\]
	Now if $S$ with $q_1, q_2 \colon S \to TI$ is any reflexive pair making the diagram on the right commute, we let $c \colon TI \to Q$ be the coequaliser of $U(q_1^\sharp)$ and $U(q_2^\sharp)$, noting that this is a reflexive pair by Lemma~\ref{lem:reflexive}.
	Then since $T$ preserves reflexive coequalisers, they are lifted by $U$, meaning that there exists a unique $T$-algebra structure $\delta \colon TQ \to Q$ making $c \colon FI \to (Q, \delta)$ a $T$-algebra homomorphism that is the coequaliser of $q_1^\sharp$ and $q_2^\sharp$.
	We also have $\lang(\mathcal{M}) \circ U(q_1^\sharp) = \lang(\mathcal{M}) \circ U(q_2^\sharp)$ by commutativity of the diagram on the right, so with $c$ coequalising $U(q_1^\sharp)$ and $U(q_2^\sharp)$ there is a unique morphism $o \colon Q \to O$ such that $o \circ c = \lang(\mathcal{M})$.
	Setting $i = c \circ \eta_I$, we have a $T$-automaton $(Q, \delta, i, o)$ with reachability map $U(i^\sharp) = c$ that accepts the language $\lang(\mathcal{M})$.
	By $\mathcal{M}$ being minimal there exists a unique $T$-automaton homomorphism $h \colon (Q, \delta, i, o) \to \mathcal{M}$.
	Then the diagram below on the left commutes.
	
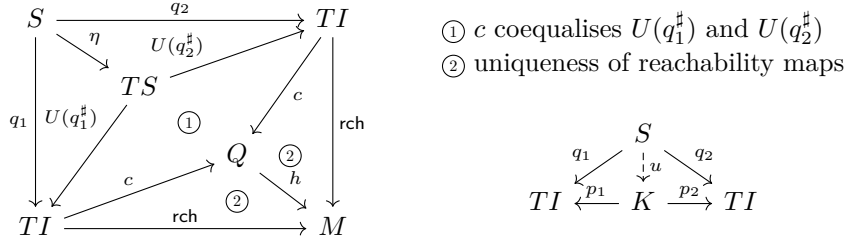
\[
		\begin{tikzcd}[column sep=.6cm,row sep=.4cm]
			S \ar{rrr}{q_2} \ar{rd}{\eta} \ar{ddd}[swap]{q_1} &
				&
				&
				TI \ar{ddd}{\reach} \ar{ddl}{c} \\
			&
				TS \ar{rru}[pos=.3]{U(q_2^\sharp)} \ar{ddl}[swap,pos=.3]{U(q_1^\sharp)} \ar[phantom]{rd}{\circled{1}} \\
			&
				&
				Q \ar{rd}{h} \ar[phantom]{r}{\circled{2}} \ar[phantom]{d}{\circled{2}} &
				{} \\
			TI \ar{rrr}{\reach} \ar{rru}{c} &
				&
				{} &
				M
		\end{tikzcd}
		\qquad
		\begin{gathered}
			\begin{array}{l}
				\circled{1}\text{ $c$ coequalises $U(q_1^\sharp)$ and $U(q_2^\sharp)$} \\
				\circled{2}\text{ uniqueness of reachability maps}
			\end{array} \\[.3cm]
			\begin{tikzcd}[column sep=.6cm,row sep=.4cm]
				&
					S \ar{ld}[swap]{q_1} \ar[dashed]{d}{u} \ar{rd}{q_2} \\
				TI &
					K \ar{l}[swap]{p_1} \ar{r}{p_2} &
					TI
			\end{tikzcd}
		\end{gathered}
	\]
	By $p_1$ and $p_2$ being the kernel pair of $\reach$ there exists a unique morphism $u \colon S \to K$ making the diagram on the right commute.
\end{proof}

\section{Tree automata with side-effects}%
\label{sec:determinisation}

We extend tree automata with various side-effects, covering as examples
non-deterministic, weighted, and non-deterministic nominal automata. The key insight is to view
them as algebras in the \emph{Kleisli category of a monad $S$}. We first recall some basic notions.

\subparagraph*{Kleisli category}
Every monad $(S,\eta,\mu)$ has an associated \emph{Kleisli category} $\Kl(S)$, whose objects are those of $\C$ and whose morphisms $X \todot Y$ are morphisms $X \to SX$ in $\C$. Given such morphisms $f \colon X \todot Y$ and $g \colon Y \todot Z$, their (Kleisli) composition $g \circ f$ is defined as $\mu_Z \circ Sg\circ f$ in $\C$. The \emph{Kleisli adjunction} $J \dashv V : \C \leftrightarrows \Kl(S)$ is given by $JX = X$, $J(f \colon X \to Y) = \eta_Y \circ f$ and $VY = SY$, $V(f \colon X \todot Y) = \mu_Y \circ Sf$.

\begin{example}
The category $\Kl(\fpow)$ has morphisms $X \to \fpow Y$, which are finitely-branching relations, and Kleisli composition is relational composition. We have that $J$ maps a function to its graph, and $V$ maps $X$ to $\fpow X$ and a relation $R \colon X \todot Y$ to the function
\[
    \lambda U \subseteq X.\{ y \mid \exists x \in U : (x,y) \in R \}.
\]
The category $\Kl(\mult_\fld)$ has morphisms $X \to \mult_\fld Y$ that are matrices over $\fld$ indexed by $X$ and $Y$ (equivalently, linear maps between the corresponding free vector spaces), and composition is matrix multiplication. The left adjoint $J$ maps a function $f \colon X \to Y$ to the matrix $Jf[x,f(x)] = 1$, for $x \in X$, and 0 elsewhere, and the right adjoint $V$ maps a matrix $X \todot Y$ to the corresponding linear function $\mult_\fld X \to \mult_\fld Y$.
\end{example}

Given an endofunctor $\Sigma$ on $\C$, a functor $\klex{\Sigma} \colon \Kl(S) \rightarrow \Kl(S)$ is an \emph{extension} of $\Sigma$ if the following diagram commutes:
\[
    \begin{tikzcd}[row sep=.35cm]
    \C \ar{d}[outer sep=1mm,swap]{J} \ar{r}{\Sigma} & \C \ar{d}[outer sep=1mm]{J} \\
    \Kl(S) \ar{r}{\klex{\Sigma}} & \Kl(S)
    \end{tikzcd}
\]
Extensions are in bijective correspondence with \emph{distributive laws} $\lambda \colon \Sigma S \Rightarrow S \Sigma$~\cite{Mulry93}, which are natural transformations satisfying certain axioms. Explicitly, we have $\klex{\Sigma} X = \Sigma X$ and $f \colon X \todot Y$ (a morphism $X \to SY$ in $\C$) is mapped to $\lambda_Y \circ \Sigma{}f \colon \Sigma X \to S \Sigma Y$, seen in $\Kl(S)$.

In~\cite[Lemma 2.4]{HasuoJS07} it is shown that a canonical distributive law in $\Set$ always exists in case $\Sigma$ is polynomial and $S$ is a commutative monad~\cite{Kock1970}.

\begin{example}%
\label{ex:distr-laws}
For $\Sigma$ a polynomial $\Set$ endofunctor, the canonical distributive law $\lambda \colon \Sigma \fpow \Rightarrow \fpow \Sigma$ can be directly defined as follows: $\lambda_X(u) = \{ v \in \Sigma X \mid (v,u) \in \img(\langle \Sigma p_1,\Sigma p_2 \rangle) \}$, where $p_1$ and $p_2$ are the left and right projections of the membership relation ${\in_X} \subseteq X \times \fpow X$.

The multiplicity monad $(\mult_\fld, e, m)$ admits a distributive law $\lambda \colon \Sigma \mult_\fld \Rightarrow \mult_\fld \Sigma$, inductively defined as follows, where $\otimes$ is the Kronecker product:
\begin{mathpar}
\lambda^\id = \id_{\mult_\fld}
\and
\lambda^A = e_A
\and
\lambda^{\Sigma_1 \times \Sigma_2} = \otimes \circ (\lambda^{\Sigma_1} \times \lambda^{\Sigma_2})
\and
\lambda^{\amalg_{i \in I} \Sigma_i} = {[\mult_\fld(\kappa_i) \circ \lambda^{\Sigma_i}]}_{i \in I}
\end{mathpar}
\end{example}
We can now define our notion of tree automaton with side-effects.

\begin{definition}[$(\Sigma,S)$-tree automaton]
Given a monad $S$, and $\klex{\Sigma} \colon \Kl(S) \to \Kl(S)$ extending a functor $\Sigma$ on $\C$, a \emph{$(\Sigma, S)$-tree automaton} is
a $\klex{\Sigma}$-tree automaton, i.e., a tuple $(Q,\delta, i, o)$, where $(Q,\delta)$ is a $\klex{\Sigma}$-algebra and $i \colon I \todot Q$ and $o \colon Q \todot O$ are morphisms in $\Kl(S)$.
\end{definition}

\begin{example} \hfill
\begin{enumerate}
	\item Let $\Gamma$ be a signature, and let $\Sigma$ be its signature functor. Then the extension of $\Sigma$ to $\Kl(\fpow)$ is obtained via the distributive law of Example~\ref{ex:distr-laws}, namely $\klex{\Sigma}X = \Sigma X$ and
	\[
            \klex{\Sigma}(f \colon X \todot Y)(\kappa_\gamma(x_1,\dots,x_{\arity(\gamma)})) = \{(\kappa_\gamma(y_1,\dots,y_{\arity(\gamma)})) \mid y_j \in f(x_j)\text{ for $1 \le j \le \arity(\gamma)$}\}
	\]
	for $\gamma \in \Gamma$.
	Non-deterministic tree automata are $(\Sigma, \pow)$-tree automata $(Q,\delta,i,o)$ with $O = \one$, the singleton set. In fact,
	we have that $\delta \colon \coprod_{\gamma \in \Gamma} Q^{\arity(\gamma)} \todot Q$ is a family of relations $\delta_\gamma \subseteq Q^{\arity(\gamma)} \times Q$; by the same token, $i \subseteq I \times Q$ relates the frontier alphabet with (possibly several) states, and $o \subseteq Q \times \one \cong Q$ is the set of final states.

	\item
    Let $\Gamma$ be a signature, and let $\Sigma$ be its signature functor.
	The extension of $\Sigma$ to $\Kl(\mult_\fld)$ is obtained via the distributive law of Example~\ref{ex:distr-laws}. Explicitly, $\klex{\Sigma}X = \Sigma X$ and $\klex{\Sigma}(f \colon X \todot Y)$ maps $\kappa_\gamma(x_1,\dots,x_{\arity(\gamma)})$ to the vector ${[\kappa_\gamma(v_z)]}_{z \in Y^{\arity(\gamma)}}$ such that $v_z$ is the $z$-th component of $f(x_1) \otimes \cdots \otimes f(x_{\arity(\gamma)})$,
	for $\gamma \in \Gamma$.
	A multiplicity tree automaton~\cite{KieferMW14} is a $(\Sigma, \mult_{\fld})$-tree automaton $(Q,\delta,i,o)$ with $I = O = \one$. In fact, we have that $\delta_\gamma$ is the \emph{transition matrix} $Q^{\arity(\gamma)} \todot Q$ in $\fld^{|Q|^{\arity(\gamma)} \times |Q|}$; similarly, $i \colon \one \todot Q$ is the \emph{initial weight vector} in $\fld^{1 \times |Q|}$, and $o \colon Q \todot \one$ is the \emph{final weight vector} in $\fld^{|Q| \times 1}$. Intuitively, $\delta_\gamma$ maps an $\arity(\gamma)$-tuple of elements of $Q$ to a linear combination over $Q$. We note that we can go beyond fields and consider $(\Sigma, \mult_\sring)$-tree automata for a semiring $\sring$, encompassing \emph{weighted} tree automata~\cite{BorchardtV03}.

    \item The nondeterministic version of the automata defined in Section~\ref{sec:nom-tree} can be obtained via the monad $\fspow \colon \Nom \to \Nom$, mapping a nominal set to the nominal set of its finitely-supported, orbit-finite subsets.\footnote{This is the finitary version of the powerset functor in $\Nom$.} This is analogous to non-deterministic nominal automata~\cite{BojanczykKL14}. We note that $\Kl(\fspow)$ is precisely the category of nominal sets and (orbit-finitely branching) equivariant relations and that $\Sigma_\lambda$ extends to relations just as a set endofunctor---the distributive law is defined as the one for $\fpow$ of Example~\ref{ex:distr-laws}, where all the maps are equivariant. A nondeterministic nominal $\Sigma_\lambda$-tree automaton is a $(\Sigma_\lambda, \fspow)$-tree automaton $(Q,\delta,i,o)$ with $O = \one$, the one-element nominal set with trivial group action. We have that $i$ and the components of $\delta$ are equivariant relations. For instance, $\delta_{\mathsf{lambda}} \subseteq [\atoms]Q \times Q$, and $o$ is an equivariant subset of $Q$.
\end{enumerate}
\end{example}

\noindent We now study language semantics of $(\Sigma,S)$-tree automata. Languages are defined via free algebras (see Section~\ref{sec:tree-aut}). It turns out that the free algebras in $\Alg(\Sigma)$ and $\Alg(\klex{\Sigma})$ are closely related. To see this, we use the following result, which follows from~\cite[Theorem 2.14]{HermidaJ98}.

\begin{lemma}\label{lm:adjunction-lift}
	Let $\Sigma \colon \C \rightarrow \C$ be a functor, $S \colon \C \rightarrow \C$ a monad,
	and $\klex{\Sigma} \colon \Kl(S) \rightarrow \Kl(S)$ an extension of $\Sigma$. Let $\lambda \colon \Sigma S \Rightarrow S \Sigma$ be the corresponding distributive law.

	Then the Kleisli adjunction $J \dashv V : \C \leftrightarrows \Kl(S)$ lifts to an adjunction in:
	$\overline{J} \dashv \overline{V}$
	\begin{mathpar}
		\begin{tikzcd}[column sep=2cm,row sep=.8cm]
			\Alg(\Sigma) \ar[shift left,start anchor=center,end anchor=center,bend left=12,shorten >=6mm,shorten <=6mm,outer sep=.3mm]{r}{\overline{J}} \ar[phantom,start anchor=center,end anchor=center]{r}{\bot} \ar{d}{} &
				\Alg(\klex{\Sigma}) \ar[shift left,start anchor=center,end anchor=center,bend left=12,shorten >=6mm,shorten <=6mm,outer sep=.3mm]{l}{\overline{V}} \ar{d}{} \\
			\C \ar[shift left,start anchor=center,end anchor=center,bend left=12,shorten >=5mm,shorten <=2mm,outer sep=.3mm]{r}{J} \ar[phantom,start anchor=center,end anchor=center]{r}{\bot} &
				\Kl(S) \ar[shift left,start anchor=center,end anchor=center,bend left=12,shorten >=2mm,shorten <=5mm,outer sep=.5mm]{l}{V}
		\end{tikzcd}
\and
	\begin{array}{l@{}l}
		\overline{J}(Q,\delta \colon \Sigma Q \to Q) &= (Q,J(\delta) \colon \klex{\Sigma} Q \todot Q)
		\\[1.5ex]
		\overline{V}(Q,\gamma \colon \klex{\Sigma} Q \todot Q) &= (SQ, V(\gamma) \circ \lambda_Q \colon \Sigma SQ \to SQ)
	\end{array}
	\end{mathpar}
\end{lemma}
From Lemma~\ref{lm:adjunction-lift}, and using that free algebras can be obtained as colimits of transfinite sequences~\cite{adamek1974free,Kelly80},
it follows that any free $\Sigma$-algebra $(\freealg{\Sigma}X,\alpha_X)$ is mapped by the left adjoint $\overline{J}$ to a free $\klex{\Sigma}$-algebra with the same carrier. Concretely, given a $\klex{\Sigma}$-algebra $(Q, \delta)$ and a morphism $i \colon I \todot Q$, a (unique) morphism $\reach \colon \freealg{\Sigma}I \to SQ$ makes the diagram on the left commute in $\C$ iff it makes the diagram on the right commute in $\Kl(S)$:
	\begin{align}\label{eq:free-alg-kleisli}
		\begin{tikzcd}[ampersand replacement=\&,column sep=1.2cm,row sep=.8cm]
			\Sigma\freealg{\Sigma}I \ar{r}{\alpha} \ar{d}[swap]{\Sigma\reach} \&
				\freealg{\Sigma}I \ar{d}[swap]{\reach} \&
				I \ar{l}[swap]{\eta} \ar{ld}{i} \&
				\klex{\Sigma}\freealg{\Sigma}I \arkl{r}{\overline{J}(\freealg{\Sigma}I, \alpha)} \arkls{d}{\klex{\Sigma}\reach} \&
					\freealg{\Sigma}I \arkls{d}{\reach} \&
					I \arkls{l}{J\eta} \arkl{ld}{i} \\
			\Sigma SQ \ar{r}{\overline{V}(Q, \delta)} \&
				SQ \&\&
				\klex{\Sigma}Q \arkl{r}{\delta} \&
					Q
			\end{tikzcd}
	\end{align}
	Note that the adjoint transpose of $\reach$ is the same morphism, seen in $\Kl(S)$.
	The functor $\overline{V}$ can be viewed as a \emph{determinisation} construction
	for $(\Sigma, S)$-tree automata.
	\begin{definition}
	Given an $(\Sigma, S)$-tree automaton $(Q,\delta,i,o)$, let $\overline{\delta} \colon \Sigma S(Q) \rightarrow S(Q)$ be the algebra structure of $\overline{V}(Q,\delta)$, i.e.,
	$(S(Q), \overline{\delta}) =
	\overline{V}(Q,\delta)$, and $\overline{o} = V(o)$.
	The $\Sigma$-tree automaton $(SQ, \overline{\delta},i,\overline{o})$ is called the \emph{determinisation} of $(Q,\delta,i,o)$.
	\end{definition}
	The following shows correctness of this determinisation construction,
	using the correspondence in~\eqref{eq:free-alg-kleisli}, and provides a concrete description of the language semantics of $(\Sigma, S)$-tree automata.

	\begin{corollary}%
		\label{cor:det-lang}
		Let $(Q,\delta,i,o)$ be a $(\Sigma,S)$-tree automaton. Then $\lang(SQ, \overline{\delta},i,\overline{o}) = \lang(Q,\delta,i,o)$.
	\end{corollary}
    We conclude this section with some example instantiations of determinisation, and of how they can be used to compute languages.
	\begin{example}\hfill
	\begin{enumerate}
		\item\label{nd-det} The determinisation of a $(\Sigma, \fpow)$-tree automaton $(Q,\delta,i,o)$ is
	\[
		\overline{\delta}_\gamma (X_1,\dots,X_k) = \bigcup_{x_1 \in X_1,\dots,x_k \in X_k} \delta_\gamma(x_1,\dots,x_k) \qquad \overline{o}(X) = \bigcup_{x \in X} o(x)
	\]
	for $\gamma \in \Gamma$ with $k = \arity(\gamma)$, and $X_1,\dots,X_k,X$ finite subsets of $Q$. This definition precisely corresponds to the usual determinisation of bottom-up tree automata (see e.g.~\cite{HabrardO06}).
    The reachability function is then given by
	\[
		\reach(t) =
		\begin{cases}
			i(t) & \text{if $t \in I$} \\
			\bigcup_{\substack{x_1 \in \reach(t_1) \\ \dots \\ x_k \in \reach(t_k)}} \delta_\gamma(x_1,\dots,x_k) & \text{if $t = (\gamma,x_1,\dots,x_k), k = \arity(\gamma)$}.
		\end{cases}
    \]
    Using Corollary~\ref{cor:det-lang},
    we have that the language of $(Q,\delta,i,o)$ is
    \[
		\lang(Q,\delta,i,o)(t) = \bigcup_{s \in \reach(t)} o(s).
	\]
	That is: a tree $t$ is accepted by $(Q,\delta,i,o)$ whenever there is a final state among those reached by parsing $t$.

	\item The determinisation of a $(\Sigma, \mult_\fld)$-tree automaton is given by
	\[
		\overline{\delta}_\gamma(\varphi_1,\dots,\varphi_k) = ( \varphi_1 \otimes \dots \otimes \varphi_k ) \bullet \delta_\gamma
		\qquad
		\overline{o}(\varphi) = \varphi \bullet o
	\]
	where $\gamma \in \Gamma$ and $\arity(\gamma) = k$; also, $\otimes$ is the Kronecker product and $\bullet$ is matrix multiplication. Explicitly, $\overline{\delta}_\gamma$ takes a $k$-tuple of vectors over $Q$ and turns it into a vector $\varphi$ over $k$-tuples of states via the distributive law (defined as the Kronecker product, see Example~\ref{ex:distr-laws}), i.e., $\varphi(q_1,\dots,q_k) = \varphi_1(q_1) \cdots \varphi_k(q_k)$, for $q_1,\dots,q_k \in Q$. The result is then multiplied by the matrix $\delta_\gamma$ to compute the successor vector.

    Similarly, the reachability function becomes
	\[
		\reach(t) =
		\begin{cases}
			i(t) & \text{if $t \in I$} \\
			(\reach(t_1) \otimes \dots \otimes \reach(t_k)) \bullet \delta_\gamma & \text{if $t = (\gamma,t_1,\dots,t_k), k = \arity(\gamma)$}.
		\end{cases}
	\]
    Hence we obtain the language $\lang(Q,\delta,i,o)(t) = \reach(t) \bullet o$, which corresponds to the language semantics given in~\cite{KieferMW14}.

    \item The case of $(\Sigma_\lambda, \fspow)$-tree automata is completely analogous to point~\ref{nd-det}. For instance,
	\[
		\overline{\delta}_{\mathsf{lambda}}(X) = \bigcup_{x \in X} \delta_{\mathsf{lambda}}(x)
	\]
	for $X$ a finitely supported orbit-finite subset of $[\atoms]Q$.
	\end{enumerate}

	\end{example}

\section{Future work}%
\label{sec:conclusion}

The algorithmic side of the iterative minimisation construction presented in Section~\ref{sec:minimisation} is left open.
For classical tree automata there exist sophisticated variants of partition refinement~\cite{HogbergMM09,AbdullaHK07},
akin to Hopcroft's classical algorithm. A generalisation to the current algebraic setting is an interesting
direction of research, for which a natural starting point would be to try and integrate in our setting the efficient coalgebraic
algorithm presented in~\cite{DorschMSW17}.

Further, we characterised the minimal automaton as the greatest fixed point of a monotone function, recovering the notion of forward bisimulations as its
post-fixed points (although it is perhaps more natural to think of these as congruences).
This characterisation suggests an integration with up-to techniques~\cite{pous:dsbook11,BonchiPPR17,BonchiP15}, which have, to the best of our knowledge,
not been applied to tree automata.
In particular, we are interested in applying these algorithms to decide equivalence of series-parallel rational and series-rational expressions~\cite{lodaya-weil-2000}.

Since completeness of Kleene Algebra is connected to minimality of deterministic finite automata~\cite{kozen-1994}, we wonder whether a completeness proof can be recovered using automata as presented in this paper.
In particular, our abstract framework might allow us to transpose such a proof to settings such as Bi-Kleene Algebra~\cite{laurence-struth-2014} or Concurrent Kleene Algebra~\cite{kappe-brunet-silva-zanasi-2018}.

\bibliographystyle{abbrv}
\bibliography{calco}

\ifarxiv%

\clearpage
\appendix

\section{Proofs for Section~\ref{sec:nerode}}

In the proofs below, we will use the following basic adjunction properties, in particular for the adjunction $F \dashv U \colon \C \leftrightarrows \EM(T)$ with adjoint transpose ${(-)}^\sharp$:
\begin{itemize}
	\item
		The transpose $f^\sharp \colon FA \to B$ for $f \colon A \to UB$ in $\C$ can be defined as $f^\sharp = y \circ Tf$, where $y$ is the $T$-algebra structure on $Y$.
	\item
		For all $f \colon X \to UY$ and $g \colon UY \to UZ$ in $\C$ we have $U(g^\sharp) \circ Tf = U({(g \circ f)}^\sharp)$.
	\item
		We have $U(\eta_X^\sharp) = \id_{TX}$.
\end{itemize}

We need the following additional lemmas in the proofs below.

\begin{lemma}\label{lem:algcomp}
	If $f \colon A \to B$ and $h \colon A \to C$ in $\EM(T)$ are such that there exists $g \colon UB \to UC$ in $\C$ with $g \circ f = h$ and $Tf$ is an epi, then $g$ is a $T$-algebra homomorphism $B \to C$.
\end{lemma}
\begin{proof}
	Let $\alpha \colon TA \to A$, $\beta \colon TB \to B$, and $\gamma \colon TC \to C$ be the respective $T$-algebra structures on $A$, $B$, and $C$.
	By commutativity of
	\[
		\begin{tikzcd}[column sep=1.1cm, row sep=.4cm]
			TA \ar{rr}{Tf} \ar{rrd}{Th} \ar{rd}[swap]{\alpha} \ar{dd}[swap]{Tf} &
				&
				TB \ar{d}{Tg} \\
			&
				A \ar{rd}{h} \ar{d}[swap]{f} &
				TC \ar{d}{\gamma} \\
			TB \ar{r}{\beta} &
				B \ar{r}{g} &
				C
		\end{tikzcd}
	\]
	and $Tf$ being an epi, we directly conclude that $g$ is a $T$-algebra homomorphism $B \to C$.
\end{proof}

\begin{lemma}\label{lem:extendnerode}
	Given a language $L$ and $q_1, q_2 \colon S \to TI$ making the diagram below on the left commute, the diagram on the right commutes.
	\begin{align*}
		\begin{tikzcd}[ampersand replacement=\&,column sep=.6cm,row sep=.4cm]
			TS \ar{rr}{T{q_2}} \ar{dd}[swap]{T{q_1}} \&
				\&
				TTI \ar{d}{\mu} \\
			\&
				\&
				TI \ar{d}{L} \\
			TTI \ar{r}{\mu} \&
				TI \ar{r}{L} \&
				O
		\end{tikzcd} &
			&
			\begin{tikzcd}[ampersand replacement=\&,column sep=.8cm,row sep=.4cm]
				TTS \ar{r}{TTq_2} \ar{d}[swap]{TTq_1} \&
					TTTI \ar{r}{T\mu} \&
					TTI \ar{d}{\mu} \\
				TTTI \ar{d}[swap]{T\mu} \&
					\&
					TI \ar{d}{L} \\
				TTI \ar{r}{\mu} \&
					TI \ar{r}{L} \&
					O
			\end{tikzcd}
	\end{align*}
	Furthermore, if $(q_1, q_2)$ is a reflexive pair, then so is $(\mu_I \circ Tq_1, \mu_I \circ Tq_2)$.
\end{lemma}
\begin{proof}
	We extend the assumption to the following commutative diagram.
	\[
		\begin{tikzcd}[ampersand replacement=\&,row sep=.4cm]
			TTS \ar{rr}{TT{q_2}} \ar{rd}{\mu} \ar{dd}[swap]{TT{q_1}} \ar[phantom,bend left=5]{drr}[pos=.6]{\circled{1}} \ar[phantom,bend right=20]{ddr}[pos=.6]{\circled{1}} \&
				\&
				TTTI \ar{r}{T\mu} \ar{d}{\mu} \&
				TTI \ar{dd}{\mu} \\
			\&
				TS \ar{r}{T{q_2}} \ar{d}[swap]{T{q_1}} \&
				TTI \ar{rd}{\mu} \ar[phantom]{ru}{\circled{2}} \&
				\\
			TTTI \ar{d}[swap]{T\mu} \ar{r}{\mu} \&
				TTI \ar{rd}{\mu} \&
				\&
				TI \ar{d}{L} \\
			TTI \ar{rr}{\mu} \ar[phantom]{ru}{\circled{2}} \&
				\&
				TI \ar{r}{L} \&
				O
		\end{tikzcd}
		\qquad
		\begin{array}{l}
			\circled{1}\text{ naturality of $\mu$} \\
			\circled{2}\text{ monad law}
		\end{array}
	\]
	As for reflexivity, $(\mu_I \circ Tq_1, \mu_I \circ Tq_2)$ is the composition of the reflexive pairs $(\mu_I, \mu_I)$ and $(Tq_1, Tq_2)$.
\end{proof}

\coequalizerlift*
\begin{proof}
	For $k \in \{1, 2\}$ we have $
		U(i^\sharp \circ q_k^\sharp) \circ \eta_A = \reach \circ U(q_k^\sharp) \circ \eta_A = \reach \circ q_k$,
	so by $U(i^\sharp)$ coequalising $q_1$ and $q_2$ we have $i^\sharp \circ q_1^\sharp = i^\sharp \circ q_2^\sharp$.
	If a $T$-algebra homomorphism $f \colon TB \to Z$ is such that $f \circ q_1^\sharp = f \circ q_2^\sharp$, then
	\[
		Uf \circ q_1 = Uf \circ U(q_1^\sharp) \circ \eta_A = Uf \circ U(q_2^\sharp) \circ \eta_A = Uf \circ q_2,
	\]
	which because $U(i^\sharp)$ coequalises $q_1$ and $q_2$ yields a unique function $u \colon UC \to UZ$ such that $u \circ \reach = f$.
	Remains to show that $u$ is a $T$-algebra homomorphism.
	Note that since $U(i^\sharp)$ is a reflexive coequaliser, $TU(i^\sharp)$ is an epi by assumption on $T$.
	Precomposing $u$ with $U(i^\sharp)$ yields the $T$-algebra homomorphism $f$, so by $TU(i^\sharp)$ being an epi and Lemma~\ref{lem:algcomp} we conclude $u$ is a $T$-algebra homomorphism $C \to Z$.
	Reflexivity of the pair follows from Lemma~\ref{lem:reflexive}.
\end{proof}

\nerodefin*
\begin{proof}
	For each subset $X \subseteq R$, we define $p_X \colon R \to TI$ by
	\[
		p_X(r) = \begin{cases}
			p_1(r) &
				\text{if $r \not\in X$} \\
			p_2(r) &
				\text{if $r \in X$}.
		\end{cases}
	\]
	We have $Tp_1 = Tp_\emptyset$ by definition.
	Consider any $t \in TR$ and let a finite $E \subseteq R$ with inclusion map $e \colon E \to R$ and $t' \in TE$ be such that $T(e)(t') = t$.
	These exist because $T$ is finitary.
	We will show by induction on $E$ that
	\begin{equation}\label{eq:subind}
		(L \circ \mu \circ Tp_\emptyset)(t) = (L \circ \mu \circ Tp_E)(t).
	\end{equation}
	The case where $E = \emptyset$ is clear, so assume $E = E' \cup \{z\}$ with $z \not\in E'$ and~\eqref{eq:subind} holds when $E'$ is substituted for $E$.
	We fix the singleton $1 = \{\square\}$ and define $d \colon R \to TI + 1$ by
	\[
		d(r) = \begin{cases}
			(\kappa_1 \circ p_1)(r) &
				\text{if $r \not\in E$} \\
			(\kappa_1 \circ p_2)(r) &
				\text{if $r \in E'$} \\
			\kappa_2(\square) &
				\text{if $r = z$},
		\end{cases}
	\]
	where $\kappa_1$ and $\kappa_2$ are the coproduct injections.
	By this definition, we have $[\id_{TI}, 1_{p_1(z)}] \circ d = p_{E'}$ and $[\id_{TI}, 1_{p_2(z)}] \circ d = p_E$, so
	\begin{align*}
		(L \circ \mu \circ Tp_\emptyset)(t) &
			= (L \circ \mu \circ Tp_{E'})(t) &
			&
			\text{(induction hypothesis)} \\
		&
			= (L \circ \mu \circ T([\id_{TI}, 1_{p_1(z)}] \circ d))(t) \\
		&
			= (L \circ \mu \circ T([\id_{TI}, 1_{p_2(z)}] \circ d))(t) &
			&
			\text{(definition of $R$)} \\
		&
			= (L \circ \mu \circ Tp_E)(t),
	\end{align*}
	thus concluding the proof of~\eqref{eq:subind}.
	Now $Tp_1 = Tp_\emptyset$ by definition and
	\[
		Tp_E(t) = T(p_E \circ e)(t') = T(p_2 \circ e)(t') = Tp_2(t),
	\]
	from which we find that
	$
		(L \circ \mu \circ Tp_1)(t) = (L \circ \mu \circ Tp_\emptyset)(t) = (L \circ \mu \circ Tp_E)(t) = (L \circ \mu \circ Tp_2)(t).
	$
	As this argument works for any $t \in TR$, we have $L \circ \mu \circ Tp_1 = L \circ \mu \circ Tp_2$.

	Now consider any set $S$ with $q_1, q_2 \colon S \to TI$ making
	\begin{equation}\label{eq:pre}
		\begin{tikzcd}[column sep=.6cm,row sep=.4cm]
			TS \ar{rr}{T{q_2}} \ar{dd}[swap]{T{q_1}} &
				&
				TTI \ar{d}{\mu} \\
			&
				&
				TI \ar{d}{L} \\
			TTI \ar{r}{\mu} &
				TI \ar{r}{L} &
				O
		\end{tikzcd}
	\end{equation}
	commute, and assume $q_1$ and $q_2$ have a common section $j \colon TI \to S$.
	We define $u \colon S \to R$ by $u(s) = (q_1(s), q_2(s))$.
	To see that this is indeed an element of $R$, note that for $k \in \{1, 2\}$,
	\begin{align*}
		L \circ \mu \circ T[\id_{TI}, 1_{q_k(s)}] &
			= L \circ \mu \circ T[\id_{TI}, q_k \circ 1_{s}] \\
		&
			= L \circ \mu \circ T[q_k \circ j, q_k \circ 1_{s}] &
			&
			\text{(section)} \\
		&
			= L \circ \mu \circ Tq_k \circ T[j, 1_{s}],
	\end{align*}
	and therefore $
		L \circ \mu \circ T[\id_{TI}, 1_{q_1(s)}] = L \circ \mu \circ T[\id_{TI}, 1_{q_2(s)}]$
	follows from~\eqref{eq:pre}.
	By definition, $u$ is the unique map making the diagram below commute.
	\begin{align*}
		\begin{tikzcd}[row sep=.4cm,ampersand replacement=\&]
			\&
				S \ar{ld}[swap]{q_1} \ar[dashed]{d}{u} \ar{rd}{q_2} \\
			TI \&
				R \ar{l}[swap]{p_1} \ar{r}{p_2} \&
				TI
		\end{tikzcd}
        \\[-\normalbaselineskip]\tag*{\qedhere}
	\end{align*}
\end{proof}

\liftnerode*
\begin{proof}
	Let $L \colon TI \to O$ be a language with Nerode equivalence $(R, p_1, p_2)$.
	Then $(p_1, p_2)$ is a reflexive pair by the Nerode equivalence property, since $(\id_{TI}, \id_{TI})$ is a reflexive pair trivially satisfying the Nerode equivalence condition.
	We apply Lemma~\ref{lem:extendnerode} to obtain from the Nerode equivalence property a unique morphism $r \colon TR \to R$ making the diagram below commute.
	\begin{equation}\label{eq:nerodestructure}
		\begin{tikzcd}[column sep=.8cm, row sep=.4cm]
			TTI \ar{d}{\mu} &
				TR \ar{l}[swap]{Tp_1} \ar{r}{Tp_2} \ar[dashed]{d}{r} &
				TTI \ar{d}{\mu} \\
			TTI &
				R \ar{l}[swap]{p_1} \ar{r}{p_2} &
				TI
		\end{tikzcd}
	\end{equation}
	We need to show that $(R, r)$ is a $T$-algebra.
	The first commutative diagram below shows that $r \circ \eta_R$ preserves $p_1$ and $p_2$, so since $\id_R$ also does this we must have $r \circ \eta_R = \id_R$ by the uniqueness property of the Nerode equivalence.
		\[
			\begin{tikzcd}[row sep=.4cm]
				&
					&
					R \ar[bend right=10]{ld}[swap]{p_1} \ar[bend left=10]{rd}{p_2} \ar{dd}{\eta}  \\
				&
					TI \ar[equals,bend right=20]{ldd}{} \ar{d}{\eta} \ar[phantom,bend left=30]{rd}[pos=.4]{\circled{2}} \ar[phantom,bend left=5]{ldd}{\circled{1}} &
					&
					TI \ar[equals,bend left=20]{rdd}{} \ar{d}[swap]{\eta} \ar[phantom,bend right=30]{ld}[pos=.4]{\circled{2}} \ar[phantom,bend right=5]{rdd}{\circled{1}} \\
				&
					TTI \ar{ld}{\mu} \ar[phantom]{rd}{\text{\small\eqref{eq:nerodestructure}}} &
					TR \ar{l}[swap]{Tp_1} \ar{r}{Tp_2} \ar{d}{r} &
					TTI \ar{rd}[swap]{\mu} \ar[phantom]{ld}{\text{\small\eqref{eq:nerodestructure}}} \\
				TI &
					&
					R \ar{ll}[swap]{p_1} \ar{rr}{p_2} &
					&
					TI
			\end{tikzcd}
			\qquad
			\begin{array}{l}
				\circled{1}\text{ monad law} \\
				\circled{2}\text{ naturality of $\eta$} \\
				\circled{3}\text{ naturality of $\mu$}
			\end{array}
		\]%
		\begin{align*}
			\begin{tikzcd}[ampersand replacement=\&]
				TTTI \ar{d}[swap]{T\mu} \ar[phantom,bend left=10]{rd}{\text{\small\eqref{eq:nerodestructure}}} \&
					TTR \ar{l}[swap]{TTp_1} \ar{r}{TTp_2} \ar{d}{Tr} \&
					TTTI \ar{d}{T\mu} \ar[phantom,bend right=10]{ld}{\text{\small\eqref{eq:nerodestructure}}} \\
				TTI \ar{d}[swap]{\mu} \ar[phantom,bend left=10]{rd}{\text{\small\eqref{eq:nerodestructure}}} \&
					TR \ar{l}[swap]{Tp_1} \ar{r}{Tp_2} \ar{d}{r} \&
					TTI \ar{d}{\mu} \ar[phantom,bend right=10]{ld}{\text{\small\eqref{eq:nerodestructure}}} \\
				TI \&
					R \ar{l}[swap]{p_1} \ar{r}{p_2} \&
					TI
			\end{tikzcd} &
				&
				\begin{tikzcd}[ampersand replacement=\&,column sep=.5cm]
					TTTI \ar{d}[swap]{T\mu} \ar{dr}{\mu} \ar[phantom,bend left=10]{rrd}[pos=.6]{\circled{3}} \&
						\&
						TTR \ar{ll}[swap]{TTp_1} \ar{rr}{TTp_2} \ar{d}{\mu} \&
						\&
						TTTI \ar{d}{T\mu} \ar{dl}[swap]{\mu} \ar[phantom,bend right=10]{lld}[pos=.6]{\circled{3}} \\
					TTI \ar{d}[swap]{\mu} \ar[phantom]{r}{\circled{1}} \&
						TTI \ar{dl}{\mu} \ar[phantom,bend right=15]{rd}[pos=.4]{\text{\small\eqref{eq:nerodestructure}}} \&
						TR \ar{l}[swap]{Tp_1} \ar{r}{Tp_2} \ar{d}{r} \&
						TTI \ar{dr}[swap]{\mu} \ar[phantom]{r}{\circled{1}} \ar[phantom,bend left=15]{ld}[pos=.4]{\text{\small\eqref{eq:nerodestructure}}} \&
						TTI \ar{d}{\mu} \\
					TI \&
						\&
						R \ar{ll}[swap]{p_1} \ar{rr}{p_2} \&
						\&
						TI
				\end{tikzcd}
		\end{align*}
		As for the other two, we use a double application of Lemma~\ref{lem:extendnerode} to see that the pair $(\mu \circ T\mu \circ TTp_1, \mu \circ T\mu \circ TTp_2)$ satisfies the Nerode equivalence conditions.
		Commutativity of the two diagrams then shows that both $r \circ Tr$ and $r \circ \mu$ are the unique map commuting with the pairs $(\mu \circ T\mu \circ TTp_1, \mu \circ T\mu \circ TTp_2)$ and $(p_1, p_2)$, so they must be equal and $(TR, r)$ is a $T$-algebra.

		It remains to show that $p_1$ and $p_2$ have a common section in $\EM(T)$.
		To this end, note that $([\eta_I, \id_{TI}], [\eta_I, \id_{TI}])$ is a reflexive pair trivially satisfying the Nerode equivalence condition.
		Thus, we obtain by the Nerode equivalence property a unique morphism $u \colon I + TI \to R$ making
		\begin{align*}
			\begin{tikzcd}[row sep=.5cm,ampersand replacement=\&]
				\&
					I + TI \ar{ld}[swap]{[\eta, \id]} \ar[dashed]{d}{u} \ar{rd}{[\eta, \id]} \\
				TI \&
					R \ar{l}[swap]{p_1} \ar{r}{p_2} \&
					TI
			\end{tikzcd}
		\end{align*}
		commute.
		Then for $k \in \{1, 2\}$,
		\[
			p_k \circ {(u \circ \kappa_1)}^\sharp = {(p_k \circ u \circ \kappa_1)}^\sharp = {(p_k \circ [\eta_I, \id_{TI}])}^\sharp = \eta_I^\sharp = \id_{(TI, \mu)}.
			\qedhere
		\]
\end{proof}

\reflexivelift*
\begin{proof}
	Assume $j \colon TB \to A$ is the common section of $q_1$ and $q_2$.
	Then, for $k \in \{1, 2\}$,
	\[
		q_k^\sharp \circ {(\eta_A \circ j \circ \eta_B)}^\sharp = U({(U(q_k^\sharp) \circ \eta_A \circ j \circ \eta_B)}^\sharp) = U({(q_k \circ j \circ \eta_B)}^\sharp) = U(\eta_B^\sharp) = \id_{TB}.
		\qedhere
	\]
\end{proof}

\fi%

\end{document}